\documentclass[10pt]{article}

\usepackage{natbib}
\usepackage[pdftex, final]{graphicx}
\DeclareGraphicsExtensions{.png,.pdf,.jpg}
\usepackage{enumerate}
\usepackage{url} 
\usepackage{soul}
\usepackage[scriptsize]{subfigure}
\usepackage[font=small]{caption}  
\usepackage{amsmath}
\usepackage{bm}
\usepackage{bbm}
\usepackage{verbatim}
\usepackage{amssymb,amsthm,amsmath}
\usepackage{dsfont}
\usepackage{mathrsfs}
\usepackage{float} 
\usepackage{booktabs} 
\usepackage{xcolor}
\usepackage{colortbl}
\usepackage{setspace}
\usepackage{authblk}
\usepackage{comment}

\usepackage{algpseudocode}
\usepackage{algorithm}

\newcommand{\indic}{\mathds{1}}
\def\E{\mathds{E}}

\def\X{\mathds{X}}

\def\R{\mathds{R}}

\newcommand{\Ecr}{\mathscr{E}}

\newcommand{\Scr}{\mathscr{S}}

\newcommand{\Ccr}{\mathscr{C}}

\DeclareMathOperator*{\argmin}{arg\,min}

\def\de{\mathrm{d}}

\newtheorem{theorem}{Theorem}

\newtheorem{corollary}{Corollary}

\def\simiid{\stackrel{\mbox{\scriptsize{iid}}}{\sim}}

\usepackage{mathtools}
\usepackage[colorlinks=true,linkcolor=blue,citecolor=blue]{hyperref}
\usepackage{cleveref}

\usepackage[top=1.2in,bottom=1.2in,left=1.2in,right=1.2in]{geometry}

\begin{document}



\title{\bf A Bayesian Nonparametric Approach to Species Sampling Problems with Ordering}

\author[,1]{Cecilia Balocchi \thanks{cecilia.balocchi@ed.ac.uk}}
\author[,2]{Federico Camerlenghi \thanks{federico.camerlenghi@unimib.it}}
\author[,3]{Stefano Favaro \thanks{stefano.favaro@unito.it}}

\affil[1]{School of Mathematics, University of Edinburgh, UK.}
\affil[2]{Department of Economics, Management and Statistics, University of Milano-Bicocca, Milano, Italy.}
\affil[3]{Department of Economics, Social Studies, Applied Mathematics and Statistics, University of Torino, Torino, Italy.}

\date{}
\maketitle

\bigskip
\begin{abstract}
Species-sampling problems (SSPs) refer to a vast class of statistical problems calling for the estimation of (discrete) functionals of the unknown species composition of an unobservable population. A common feature of SSPs is their invariance with respect to species labeling, which is at the core of the Bayesian nonparametric (BNP) approach to SSPs under the popular Pitman-Yor process (PYP) prior. In this paper, we develop a BNP approach to SSPs that are not “invariant” to species labeling, in the sense that an ordering or ranking is assigned to species’ labels. Inspired by the population genetics literature on age-ordered alleles' compositions, we study the following SSP with ordering: given an observable sample from an unknown population of individuals belonging to species (alleles), with species' labels being ordered according to weights (ages), estimate the frequencies of the first $r$ order species’ labels in an enlarged sample obtained by including additional unobservable samples. By relying on an ordered PYP prior, we obtain an explicit posterior distribution of the first $r$ order frequencies, with estimates being of easy implementation and computationally efficient. We apply our approach to the analysis of genetic variation, showing its effectiveness in estimating the frequency of the oldest allele, and then we discuss other potential applications.
\end{abstract}


\maketitle


\section{Introduction} \label{sec:intro}

Species sampling problems (SSPs) refer to a vast class of statistical problems, of which the estimation of the number of unseen species is arguably the most popular example \citep{Goo(56),Efr(76),Lij(07),Orl(17)}. Consider $n\geq1$ observable samples from a generic population of individuals, with each individual taking a value in a (possibly infinite) discrete space of symbols or species' labels. The unseen-species problem assumes that observable samples are modeled as a random sample $(X_{1},\ldots,X_{n})$ from an unknown discrete distribution $p$, and calls for estimating 
\begin{equation}\label{unseen_ssp}
K_{m}^{(n)}=|\{X_{n+1},\ldots,X_{n+m}\}\setminus \{X_{1},\ldots,X_{n}\}|,
\end{equation}
namely the number of hitherto unseen symbols that would be observed if $m\geq1$ additional samples $(X_{n+1},\ldots,X_{n+m})$ were collected from the same distribution. SSPs comprise of generalizations or refinements of the unseen-species problem, calling for the estimation of (discrete) functionals of the species' composition of unobservable samples, e.g. missing mass, discovery probabilities, unseen species with prevalences and coverages of prevalence. We refer to \citet{Den(19)} and \citet{Bal(22)} for reviews of SSPs, both in methods and applications, mostly in the field of biological sciences but also in machine learning, electrical engineering, computer science and information theory.

A common feature of SSPs is that species' labels identifying the $X_{i}$'s are immaterial in the definition of the functional of interest, as for instance in \eqref{unseen_ssp}, thus making SSPs ``invariant" to species labeling. Such a feature is at the core of the Bayesian nonparametric (BNP) approach to SSPs \citep{Lij(07),Lij(08),Fav(09),Fav(13)}, which relies on the specification of a (nonparametric) prior $\mathscr{P}$ for the unknown distribution $p$, i.e.
\begin{align}\label{ex_model}
X_i\,|\,P & \quad\simiid\quad P \qquad i=1,\ldots,n,\\
\notag P& \quad\sim\quad\mathscr{P}.
\end{align}
Species sampling models (SSMs) \citep{Pit(96)} provide a natural choice for the prior distribution $\mathscr{P}$, including the celebrated Dirichlet process (DP) prior \citep{Fer(73)} and the Pitman-Yor process (PYP) prior \citep{Pit(97)}. Under the BNP model \eqref{ex_model} with a SSM for $P$, the random sample $(X_{1},\ldots,X_{n})$ induces a random partition $\tilde{\pi}_{n}$ of $[n]=\{1,\ldots,n\}$ whose blocks correspond to the (equivalence) classes induced by the equivalence relation $i\sim j\iff X_{i}=X_{j}$ almost surely. In particular, $\tilde{\pi}_{n}$ is exchangeable \citep[Chapter 2]{Pit(06)}, namely its distribution is such that the probability of any partition of $[n]$ with $k$ blocks of frequencies $(n_{1},\ldots,n_{k})$ is a symmetric function of compositions $(n_{1},\ldots,n_{k})$ of $[n]$. The exchangeability of $\tilde{\pi}_{n}$ implies that blocks' labels are immaterial, and therefore it legitimates the BNP approach to SSPs under the class of SSMs.

\subsection{Our contributions}

In this paper, we consider SSPs that are not ``invariant" with respect to species labeling, in the sense that an ordering or ranking is assigned to species' labels, and we develop a BNP approach to such problems. Under the infinitely-many neutral alleles model for the evolution of genetic populations \citep{Ewe(72),Kin(75),Wat(77),Gri(79)}, the work of \citet{Don(86)} first investigated the alleles' composition of a random sample from the population by also taking into account the ages of alleles, namely the times elapsed since the first time each allele first appeared in the sample. This study led to the introduction of an age-ordered version of the random partition induced by the DP prior, where species are alleles and species' labels are ordered according to the age of alleles in such a way that the smaller the order the older the allele. Besides providing distributional properties of the age-ordered random partition, \citet{Don(86)} applied such a model to answer a critical question raised in \citet{Cro(72)}: ``Is the most frequent allele the oldest?". Under the infinitely-many neutral alleles model, \citet{Don(86)} came up with a positive answer to such a question, showing that the probability that an allele represented $i$ times in sample of size $n$ is the oldest is $i/n$ \citep{Kel(77),Wat(77)}. The question of \citet{Cro(72)} is to some extent a SSP with ordering, as the object of interest involves a species' label with a precise order, namely the species' label of order $1$ that corresponds to the oldest allele.

Inspired by the seminal work of \citet{Don(86)}, we study the following SSP with ordering: assuming $n$ observable samples to be modeled as a random sample $((T_{1},X_{1}),\ldots,(T_{n},X_{n}))$ from an unknown discrete distribution $q$ on $\R_+ \times \X$, with the positive $T_{i}$'s being considered as weights inducing an ordering among species' labels identifying the $X_{i}$'s, on a general (measurable) space $\X$, we estimate the frequencies of the first $r$ order species' labels in an enlarged sample obtained by including $m$ additional unobservable samples $((T_{n+1},X_{n+1}),\ldots,(T_{n+m},X_{n+m}))$ from the same unknown distribution. Within the population genetic setting of \citet{Don(86)}, this ordered SSP corresponds to the estimation of the frequency of the $r$ oldest alleles in a sample of size $(n+m)$ based on $n$ observable samples. We introduce a BNP approach to estimate the first $r$ order frequencies, which relies on the use of the class of spatial neutral to the right SSMs \citep{Jam(06)}, or ordered SSMs, as prior distributions for $q$. The most popular ordered SSM is the ordered DP prior, which is known to induce the age-ordered random partition of \citet{Don(86)}. See also \citet{Gne(05)} and references therein. Here, we consider the more general ordered PYP prior  \citep{Gne(05),Jam(06)}, and we determine the posterior distribution of the first $r$ order frequencies; then, a BNP estimator is proposed in terms of the posterior mean, whose closed-form expression results to be of easy implementation and also computationally efficient. Of special interest is the case $r=1$ that, in the original setting of \citet{Don(86)}, leads to an estimate of the frequency of the oldest allele.

We present an empirical validation of the effectiveness of our BNP approach, both on synthetic data and real data. It is natural to focus on applications to genetic data, for which the weights $T_{i}$'s have an interpretation as the ages of the alleles $X_{i}$'s. The problem of modeling the interplay between the alleles' composition of a genetic population and the age of alleles dates back to the 1970s and the 1980s, and nowadays the genealogical structure of alleles is well recognized as a fundamental aspect in many inferential (decision) processes in the field of population genetics. In particular, investigating genetic variation while incorporating the information on the variants' age enhances the investigation of several problems, such as analyzing and comparing populations structure, detecting which samples are related, studying demographic history, and learning about genetic susceptibility to disease \citep{Mat(14)}. For example, it enables researchers to use variants' age distribution to compare populations, to differentiate age distributions in pathogenic and benign variants, and to learn about genealogical history \citep{Alb(20)}. Here, we apply our BNP methodology to the problem of estimating the frequency of the oldest allele, using genetic variation data from the 1000 Genomes Project \citep{10002015global} and variants' age estimates from the Human Genome Dating Project \citep{Alb(20)}. Thanks to our posterior estimator, we can not only answer inferential questions on the frequency of the oldest allele in an observed sample, but also make predictions by analyzing an enlarged sample. By studying the trajectory of this frequency as a function of the enlarged sample size, we can enhance our understanding of the population distribution, thereby addressing the investigation of the aforementioned issues more effectively.

Besides population genetics, SSPs with ordering arise in at least other two contexts: i) citations to academic articles and ii) online purchases of items. In the context of citations to academic articles, with articles being ordered according to their publication's dates, one may be interested in the frequency of citations to the oldest paper. Citation data are often analyzed in the framework of citation networks to study the movement of ideas in academic fields or examine scholars' influence \citep{Por(17)}. Incorporating knowledge of articles' age permits the investigation of the effects of time on the number of citations \citep{Haj(05)} and to answer questions such as ``are older papers more frequently cited than newer papers?''. Assuming that each citation represents an observation, that the article cited represents a species, and that the order of the cited article is determined by its publication's date, we may apply our BNP approach in order to predict the frequency of citations to the oldest article in future observations. In the context of online purchases of items, species' labels are represented by the items purchased, with items being ordered according to their costs, as well as by a generic independent measure of popularity. In such a context, we may apply our BNP approach to study the distribution of the most popular item in future purchases, which is particularly relevant in order to plan suitable changes in the current marketing strategies.

\subsection{Organization of the paper}

The paper is structured as follows. In Section \ref{sec2} we present the ordered PYP prior and review its sampling structure in terms of sampling formulae and predictive distributions. In Section \ref{sec3} we provide the posterior distribution of the first $r$ order frequencies, with emphasis on the special case $r=1$, and obtain corresponding estimators. Section \ref{sec4} contains numerical illustrations of our BNP approach, both on synthetic and real data, whereas in Section \ref{sec5} we discuss our work and some directions for future research. Additional numerical illustrations on genetic data, an illustration in the context of citations to academic articles, and the proofs of our results are deferred to the Supplementary Material.


\section{The ordered PYP}\label{sec2}

To introduce the ordered PYP, it is useful to recall the PYP and its sampling structure. Let $P$ be a PYP with parameter $\alpha\in[0,1)$ and $\theta>-\alpha$ on a measurable space $\mathbb{X}$. That is $P=\sum_{i\geq1}P_{i}\delta_{S_{i}}$, where: i) $P_{1}=V_{1}$ and $P_{i}=V_{i}\prod_{1\leq j\leq i-1}(1-V_{j})$ with $(V_{i})_{i\geq1}$ being independent Beta random variables with parameter $(1-\alpha,\theta+i\alpha)$; ii) $(S_{i})_{i\geq1}$ be random variables, independent of the $V_{i}$'s, and independent and identically distributed according to a non-atomic distribution $\nu$ on $\mathbb{X}$ \citep{Per(92),Pit(95)}. Because of the (almost sure) discreteness of $P$, a random sample $(X_{1},\ldots,X_{n})$ from $P$ induces a random partition $\tilde{\pi}_{n}$ of $[n]$ into $K_{n}\leq n$ blocks, labelled by $\{X_{1}^{\ast},\ldots,X_{K_{n}}^{\ast}\}$, with frequencies $N_{j,n}=|i\in[n]\text{ : }X_{i}=X_{j}^{\ast}|$ for $j=1,\ldots,K_{n}$ and such that $N_{j,n}\geq1$ and $\sum_{1\leq j\leq K_{n}}N_{j,n}=n$. In particular, if we set $(a)_{(r)}=\prod_{0\leq i\leq r-1}(a+i)$ for any $a\geq0$ and $r\in\mathbb{N}_{0}$, then the probability of any partition of $[n]$ with $k$ blocks of frequencies $(n_{1},\ldots,n_{k})$ is
\begin{equation}\label{eq_eppf}
\Pi_{k}^{(n)}(n_{1},\ldots,n_{k})=\frac{\prod_{i=1}^{k}(\theta+(i-1)\alpha)}{(\theta)_{(n)}}\prod_{i=1}^{k}(1-\alpha)_{(n_{i}-1)}.
\end{equation}
Equation \eqref{eq_eppf} is referred to as the exchangeable partition probability function (EPPF), a concept introduced in \citet{Pit(95)} as a development of results in \citet{Kin(78)}. For $\alpha=0$ the PYP reduces to DP, and hence \eqref{eq_eppf} reduces to the Ewens sampling formula \citep{Ewe(72)}. See \citet[Chapter 3 and Chapter 4]{Pit(06)} for a detailed account of EPPFs.

The predictive distribution of the PYP provides a generative scheme for the random partition $\tilde{\pi}_{n}$. This is typically stated in terms of the Chinese Restaurant Process \citep[Chapter 3]{Pit(06)}, which is a sequential construction of $\tilde{\pi}_{n}$ through the metaphor of customers (observations) sitting at tables (species) of a restaurant. Under Chinese Restaurant Process, the first customer $X_{1}$ arrives and is assigned to a table. After $n$ customers $(X_{1},\ldots,X_{n})$ have arrived and have been assigned to $k$ tables $\{X_{1}^{\ast},\ldots,X_{k}^{\ast}\}$, with $n_{i}$ being the number of customers at table $i=1,\ldots,k$, the customer $X_{n+1}$ arrives and
\begin{itemize}
\item[i)] she will sit at a (``new") table $X^{\ast}$, that is a table not already occupied, with a probability
\begin{equation}\label{eq:prednew_class}
p^{\tiny{(new)}}=\frac{\theta+k\alpha}{\theta+n};
\end{equation}
\item[ii)] she will sit at a table $X^{\ast}_{j}$ that has been already occupied, for $j=1,\ldots,k$, with a probability
\begin{equation}\label{eq:predold_class}
p^{\tiny{(old)}}_{j} =\frac{n_{j}-\alpha}{\theta+n}.
\end{equation}
\end{itemize}
We refer to \citet[Chapter 3 and Chapter 4]{Pit(06)} for a detailed account of Chinese Restaurant Process and its generalizations to SSMs. In particular, the PYP is characterized as the sole SSM for which $p^{\tiny{(new)}}$ depends only on $(n,k)$ and $p^{\tiny{(old)}}_{j}$ depends only on $(n,n_{j})$ 
(\cite{Zab(97)}; see \cite{Bac(17)} for more general sufficiency postulates).

Equation \eqref{eq_eppf} is a symmetric function of compositions $(n_{1},\ldots,n_{k})$ of $[n]$, that is the random partition $\tilde{\pi}_{n}$ induced by the PYP is an exchangeable random partition \citep[Chapter 2]{Pit(06)}. The ordered PYP is a discrete random probability measure generalizing the PYP, in the sense that random sampling from the ordered PYP allows to couple each species' label with a corresponding order \citep{Gne(05),Jam(06)}. 
An ordered PYP $Q$ is an almost surely discrete random probability measure with parameters $\alpha\in[0,1)$ and $\theta>0$, that can be defined relying on the de Finetti theorem. Indeed, if  $\{(T_i,X_i)\}_{i \geq 1}$ is an exchangeable sequence of observations whose directing measure is an ordered PYP $Q$, i.e., $(T_{i},X_i)\,|\,Q  \simiid Q$ as $i \geq 1$, we can characterize $Q$ by assigning the predictive distributions of the associated exchangeable sequence.
In order to do this, consider the random sample $((T_{1},X_{1}),\ldots,(T_{n},X_{n}))$ from $Q$, with the $T_{i}$'s being viewed as weights that induce an ordering among species' labels identifying the $X_{i}$'s.
Because of the (almost sure) discreteness of $Q$, the random sample $((T_{1},X_{1}),\ldots,(T_{n},X_{n}))$ from $Q$ induces a random partition of $[n]$ into $K_{n}\leq n$ blocks, labelled by a $K_{n}$-tuple $((T^{\ast}_{1},X_{1}^{\ast}),\ldots,(T^{\ast}_{K_{n}},X_{K_{n}}^{\ast}))$ that is ordered according to the $T^{\ast}_{j}$'s in such a way that $T^{\ast}_{1}>\cdots> T^{\ast}_{K_{n}}$, with corresponding ordered frequencies $M_{j,n}=|{i\in[n]\text{ : }(T_{i},X_{i})=(T^{\ast}_{j},X^{\ast}_{j})}|$ for $j=1,\ldots,K_{n}$ and such that $M_{j,n}\geq1$ and $\sum_{1\leq j\leq K_{n}}M_{j,n}=n$. Species' labels $X_{j}^{\ast}$'s are thus ordered with respect to the decreasing ordering of the weights $T^{\ast}_{j}$'s, namely the larger the weight the smaller the order, such that $M_{j,n}$ is the frequency of the species' label of order $j$ that corresponds to the $j$-th largest weight $T^{\ast}_{j}$. An analogous construction follows for ordered SSMs \citep{Jam(06)}. 

In analogy with the Chinese Restaurant Process, the predictive distribution of the ordered PYP $Q$ may be stated as an ordered version of the Chinese Restaurant Process, with tables ordered according to weights \citep{Jam(06)}. In particular, under the ordered Chinese Restaurant Process, the first $n$ customers $((T_{1},X_{1}),\ldots,(T_{n},X_{n}))$ arrive and they are assigned to the $k$ ordered tables $((T^{\ast}_{1},X_{1}^{\ast}),\ldots,$ $(T^{\ast}_{k},X_{k}^{\ast}))$, with the table of order $j$ corresponding to the $j$-th largest weight $T^{\ast}_{j}$. Hereinafter, we denote by $m_{j}$ the number of customers seated at the table of order $j$ for $j=1,\ldots,k$, and we set $r_{j}=m_j + \cdots +m_k$, for any $j=1,\ldots,k$, and $r_{k+1}=0$. Then, the customer $(T_{n+1},X_{n+1})$ arrives and
\begin{itemize}
\item[i)] she will sit at a ``new" table  $(T^{\ast},X^{\ast})$ of order $j=1,\ldots,k+1$, that is a table not already occupied and whose order $j$ is determined through the weight $T^{\ast}$, with a probability
\begin{equation}\label{eq:prednew}
q^{\tiny{(new)}}_{j} = \frac{\theta+\alpha r_{j}}{(1+r_{j} ) (\theta+n)} \prod_{i =1 }^{j-1}   \frac{r_{i} (\alpha  r_{i+1} +\alpha +\theta m_{i})}{(r_{i}+1) (\alpha  r_{i+1} +\theta m_{i})},
\end{equation}
where $T^*$ and $X^*$ are generated from two non-atomic distributions on $\R_+$ and $\X$ respectively;
\item[ii)] she will sit at a table $(T^{\ast}_{j},X^{\ast}_{j})$ that has been already occupied, for $j=1,\ldots,k$, with a probability
\begin{equation}\label{eq:predold}
q^{\tiny{(old)}}_{j} = \frac{r_{j}(m_j-\alpha) (\alpha r_{j+1} +\theta m_j +\theta)}{(1+r_{j})(\theta+n)(\alpha r_{j+1} +\theta m_j )}\prod_{i = 1}^{j-1} \frac{r_{i} (\alpha r_{i+1}+\alpha+\theta m_{i})}{(r_{i}+1) (\alpha r_{i+1}+\theta m_{i})}.
\end{equation}
\end{itemize}
Given that the $T^*_j$ only affect the distribution of an ordered partition through the ordering induced on the clusters, we avoid using specific notation for its distribution, as it is immaterial. 
We refer to \citet{Gne(05)} and \citet{Jam(06)} for a detailed account of \eqref{eq:prednew} and \eqref{eq:predold}. The predictive distribution of the ordered DP arises from \eqref{eq:prednew} and \eqref{eq:predold} by setting $\alpha=0$.

Equation \eqref{eq:prednew} and Equation \eqref{eq:predold} provide a generative scheme for the random partition of $[n]$ induced by the ordered PYP $Q$. That is: i) if the $(n+1)$-th customer $(T_{n+1},X_{n+1})$ sits at the new table $(T^{\ast},X^{\ast})$, which happens with probability \eqref{eq:prednew}, then the order of such a table with respect to the ordering of the already occupied tables $((T^{\ast}_{1},X_{1}^{\ast}),\ldots,(T^{\ast}_{k},X_{k}^{\ast}))$ is determined by $T^{\ast}$, thus possibly changing the ordering of occupied tables by shifting the order of tables $(T^{\ast}_{j},X^{\ast}_{j})$'s with weights smaller than $T^*$; ii) if the $(n+1)$-th customer $(T_{n+1},X_{n+1})$ sits at a table $(T^{\ast}_{j},X^{\ast}_{j})$ that is already occupied, which happens with probability \eqref{eq:predold}, then the order of such a table with respect to the ordering of the already occupied tables is determined by $T^{\ast}_{j}$, thus not changing the ordering of occupied tables. In other terms, a new customer sitting at a new table may determine a change in the ordering of the occupied tables, whereas a new customer sitting at a table occupied does not determine a change in the ordering of the already occupied tables.

\citet{Gne(05)} and \citet{Jam(06)} first investigated properties of the random partition induced by the ordered PYP, and introduced the notion of ordered EPPF. Generalizing the definition of EPPF, the ordered EPPF is defined as the probability of any ordered partition of $[n]$ with $k$ blocks of frequencies $(m_{1},\ldots,m_{k})$. Here the term \textit{ordered partition} refers to a partition of $[n]$, where the blocks are ordered in accordance with the weights $T_j$. \citet{Gne(05)} showed that the ordered PYP induces a random partition whose ordered EPPF is
\begin{equation}\label{eq_eppf_ord}
\Phi^{(n)}_{k}(m_{1},\ldots,m_{k})=\frac{\prod_{i=1}^{k}\frac{\theta m_{i}+\alpha r_{i+1}}{r_{i}}}{(\theta)_{(n)}}\prod_{i=1}^{k}(1-\alpha)_{(m_{i}-1)}.
\end{equation}
See also \citet{Gne(10)}, and references therein, for a comprehensive account on ordered EPPFs and generalizations thereof. Note that the EPPF \eqref{eq_eppf} can be recovered from the ordered EPPF \eqref{eq_eppf_ord} by summing over the set $S_{k}$ of all possible permutations of the $k$ blocks, that is
\begin{equation}\label{eq:interpl}
\Pi_{n}^{(k)}(m_{1},\ldots,m_{k})=\sum_{\pi\in S_{k}}\Phi^{(n)}_{k}(m_{\pi(1)},\ldots,m_{\pi(k)}).
\end{equation}
See Section S1.1 of the Supplementary Material for details on Equation \eqref{eq:interpl}. The distribution of the age-ordered partition of \citet{Don(86)} arises from \eqref{eq_eppf_ord} by setting $\alpha=0$, where species' labels are ordered according to weights $T_{i}$'s that are interpreted as the ages of alleles. Another special case of the ordered EPPF \eqref{eq_eppf_ord} is obtained by setting $\alpha\in(0,1)$ and $\theta=0$. See \citet{Fav(16)} and references therein for details.

By applying the ordered EPPF \eqref{eq_eppf_ord}, one may compute the probability $P_{n}(i;\alpha,\theta)$ that a species with frequency $i$ has species' label of order $1$, i.e. the species' label corresponding to the largest weight $T^{\ast}_{1}$. For $\alpha=0$, \citet{Don(86)} computed such a probability, showing that it is independent of $\theta$ and also an increasing (linear) function of $i$, i.e. 
\begin{equation}\label{oldest_dp}
P_{n}(i;0,\theta)=\frac{i}{n}.
\end{equation}
Within the population genetic setting of \citet{Don(86)}, Equation \eqref{oldest_dp} shows that the most frequent allele is the oldest allele. In general, for any $\alpha\in[0,1)$ and $\theta>0$ it holds
\begin{equation}\label{eq:Pn}
P_{n}(i;\alpha,\theta)=\frac{\alpha n +i (\theta-\alpha)}{n}\E \left[ \frac{1}{\theta +\alpha K_{n-i}}  \right],
\end{equation}
where $K_{n-i}$ is the number of distinct species in $(n-i)$ random samples for the ordered PYP $Q$, with the proviso $K_{0}=0$ \citep[Chapter 3]{Pit(06)}. See Section~S1.2 of the Supplementary Material for the proof of Equation \eqref{eq:Pn}. It is easy to show that \eqref{eq:Pn} reduces to the probability $\eqref{oldest_dp}$ for $\alpha=0$.  The comparison between \eqref{eq:Pn} and \eqref{oldest_dp} is critical, as it highlights the increased flexibility of the ordered PYP compared to the ordered DP ($\alpha=0$). In fact, differently from the probability \eqref{oldest_dp}, the probability \eqref{eq:Pn} depends on $(\alpha,\theta)$ and, most importantly, it is no more an increasing (linear) function of $i$. Figure \ref{fig:grafico_i} shows that the probability \eqref{eq:Pn} may increase or decrease in $i$ according to the value of $(\alpha,\theta)$; for instance, for $\alpha\in(0,1)$ and $\theta=0$ the probability \eqref{eq:Pn} is the product of a term decreasing in $i$ and one increasing in $i$. The non-increasing behavior of $P_n(i;\alpha,\theta)$ for $\alpha > \theta$ and for $\theta = 0$ is depicted in Figure~S1 of the Supplementary Material.

\begin{figure}[h!]
  \begin{center}
    \subfigure[]{\includegraphics[width=0.48\linewidth]{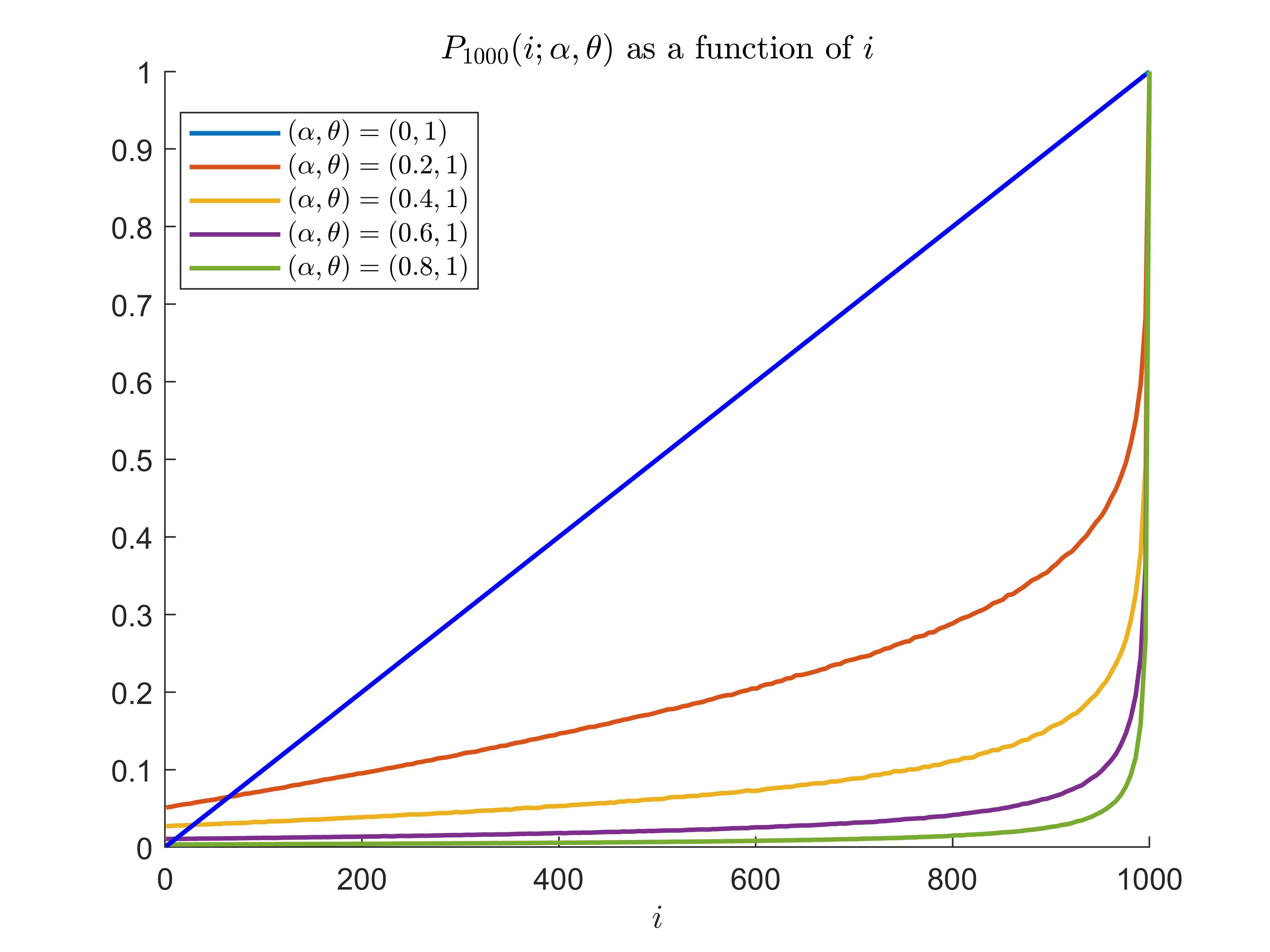} \label{fig:theta1_grafico_i}}
    \subfigure[]{\includegraphics[width=0.48\linewidth]{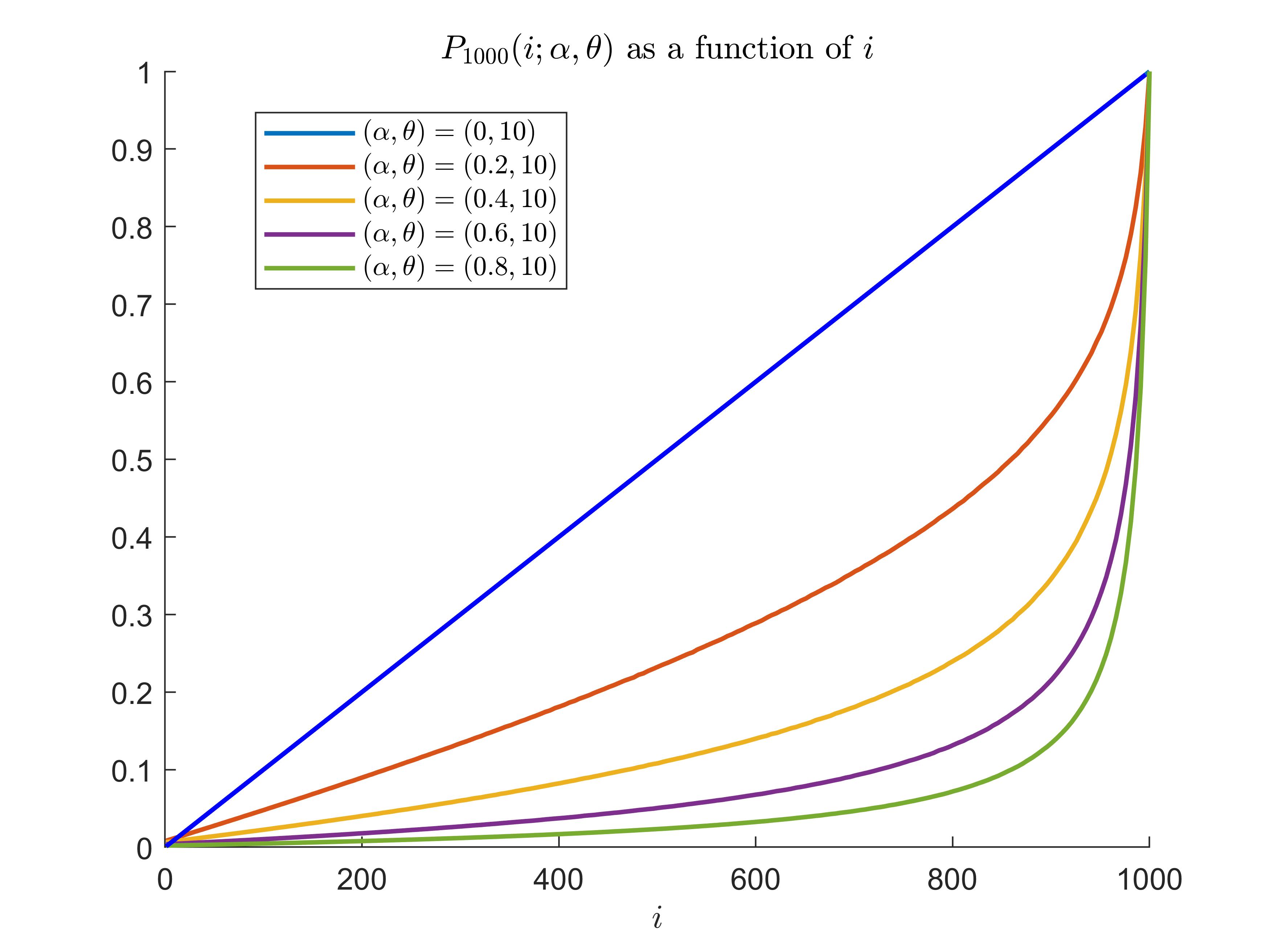} \label{fig:theta10_grafico_i}}\\
    \subfigure[]{\includegraphics[width=0.48\linewidth]{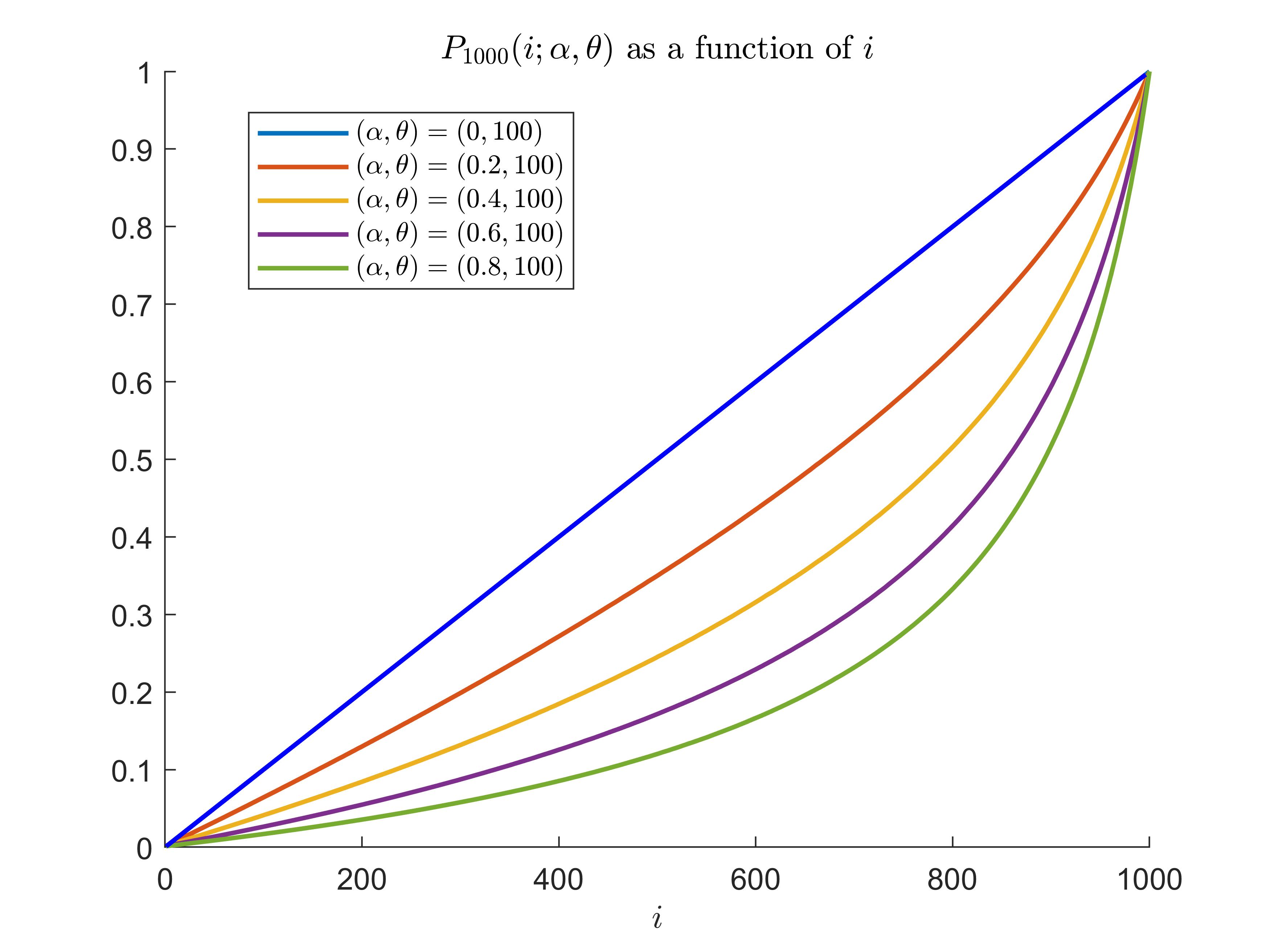} \label{fig:theta100_grafico_i}}
    \subfigure[]{\includegraphics[width=0.48\linewidth]{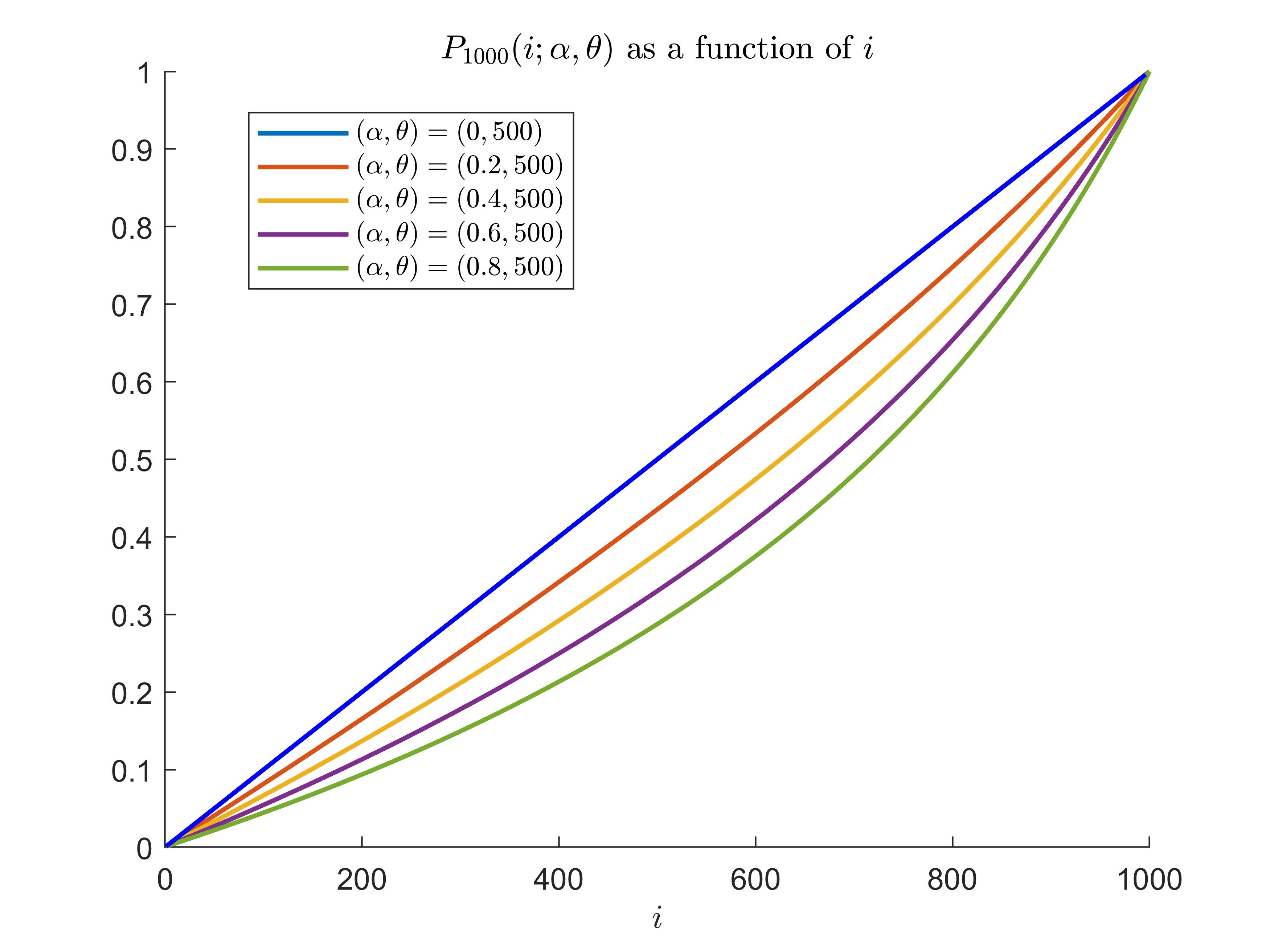} \label{fig:theta500_grafico_i}}\\
  \end{center}
  \begin{flushleft} 
    \caption{The probability $P_n(i ; \alpha, \theta)$, where $n=1000$, as a function of the frequency $i$, for different values of 
    $\alpha$ and $\theta$. Each panel corresponds to a different $\theta$, from top left to bottom right, we have: $\theta=1,10, 100, 500$.}
    \label{fig:grafico_i}
  \end{flushleft}
\end{figure}

To conclude, it is worth mentioning the construction of the order PYP $Q$ that induces the predictive distributions \eqref{eq:prednew} and \eqref{eq:predold}, though we will not make use of such a construction in the paper. The ordered PYP belongs to the class of spatial neutral to the right SSMs defined by \cite{Jam(06)}, and hence it is defined as a discrete random probability measure on the product space $\Scr= \R^+ \times \X$, from which the observations $(T_i, X_i)$, as $i =1, \ldots , n$, are sampled. To formalize such a definition, we consider a marked Poisson process $N$ \citep{Kin(93)} on the space $[0,1]\times \Scr$ with mean intensity given  by 
\[
\nu (\de u , \de s , \de x) := \rho (\de u | s) \Lambda_0(\de s , \de x)
\]
where $\rho$ is a L\'evy density, while $\Lambda_0$ is a hazard measure on the space $\Scr$. Thus, one may define a functional of the Poisson process $N$ as follows  $\Lambda (\de s, \de x) = \int_0^1 u N (\de u, \de s, \de x)$, which turns out to be a completely random measure \citep{Dal(08)}; $\Lambda$ represents a hazard measure in the framework of survival analysis. Now, define the survival function associated with $\Lambda$ as $-\log (S(t-)):= \int_{[0,1]\times \Scr} [-\indic_{\{s < t\}} \log (1-u)] N (\de u, \de s, \de x)$.
Then, a spatial neutral to the right random probability measure equals $Q(\de t, \de x) := S(t-) \Lambda (\de t, \de x)$. By choosing $\rho$ as in \citep[Section 6.2]{Jam(06)}, the law of the resulting $Q$ is the de Finetti measure associated with the prediction rules \eqref{eq:prednew}--\eqref{eq:predold}. 
Note that in this construction the $T_i$'s are considered as times, but they can be seen more generally as weights inducing an order, making the model more widely applicable.
See \cite{Jam(06)} for general properties of the ordered PYP $Q$, including the posterior distribution.


\section{BNP inference for the first $r$ order frequencies}\label{sec3}

In analogy with SSPs, SSPs with ordering assume $n\geq1$ observable samples from a population of individuals, with each individual taking a value in a (possibly infinite) discrete space of symbols, and then consider $m\geq1$ additional unobservable sample from the same population. The critical difference between SSPs and SSPs with ordering lies in the definition of the (discrete) functional of interest: while in SSPs such a functional is ``invariant" with respect to species ordering and deals with the species' composition of the additional samples, in SSPs with ordering the functional is not ``invariant" with respect to species ordering and deals with the species' composition of both the additional samples and the enlarged sample. The estimation of the first $r$ order frequencies is arguably the most natural example of SSPs with ordering. Assuming $n$ observable samples to be modeled as a random sample $((T_{1},X_{1}),\ldots,(T_{n},X_{n}))$ from the ordered PYP $Q$:
\begin{align}\label{ex_model_ord}
(T_{i},X_i)\,|\,Q & \quad\simiid\quad Q \qquad i=1,\ldots,n,
\end{align}
i.e., the observations are updated according to the predictive laws \eqref{eq:prednew}-\eqref{eq:predold}. We introduce a BNP approach to estimate the frequencies of the first $r$ order species in an enlarged sample obtained by collecting $m$ additional samples $((T_{n+1},X_{n+1}),\ldots,(T_{n+m},X_{n+m}))$ from the same $Q$.

\subsection{Posterior distributions for the first $r$ order frequencies}

We start by introducing a marginal distribution related to the random partition induced by an ordered PYP $Q$, with parameters $(\alpha,\theta)$. For any $n\geq1$ let $((T_{1},X_{1}),\ldots,(T_{n},X_{n}))$ be a random sample under the BNP model \eqref{ex_model_ord}, such that the sample features $K_{n}=k$ distinct species with ordered frequencies $\mathbf{M}_{n}=\mathbf{m}$. Hereinafter, for the sake of simplicity in notation, we denote by $|\mathbf{m}|_{1:r}$ the sum of the first $r$ elements of $\mathbf{m}$, i.e. $|\mathbf{m}|_{1:r}=\sum_{1\leq j\leq r}m_{j}$, with the proviso $|\mathbf{m}|_{1:0}=0$. For any index $r \in \{  1, \ldots , n\}$ such that $ r \leq |\mathbf{m}|_{1:r} \leq n-k+r$, if we set
\begin{displaymath}
C_{r,n}(\alpha,\theta,\mathbf{m})=\prod_{j=1}^{r}\frac{  [\alpha (n-|\mathbf{m}|_{1:j} )+ \theta m_j]}{n-|\mathbf{m}|_{1:j-1}}(1-\alpha)_{(m_j -1)}
\end{displaymath}
then 
\begin{align} \label{eq:dist_a_priori}
& \text{Pr}[M_{1,n} = m_1 , \ldots , M_{r,n} = m_r ,  K_n \geq r]\\
&\notag\quad= {n\choose m_{1},\ldots,m_{r},n-|\mathbf{m}_{1:r}|}\frac{C_{r,n}(\alpha,\theta,\mathbf{m})}{(\theta+n-|\mathbf{m}|_{1:r})_{|\mathbf{m}|_{1:r}}}.
\end{align}
See Section~S1.4 of the Supplementary Material for the proof of Equation \eqref{eq:dist_a_priori}. Equation \eqref{eq:dist_a_priori} generalizes \citet[Proposition 6.1]{Don(86)}, which is recovered from \eqref{eq:dist_a_priori} by letting $\alpha \to 0$. For $r=1$, Equation \eqref{eq:dist_a_priori} provides the distribution of first order frequency, i.e.
\begin{equation} \label{eq_old}
\text{Pr}[M_{1,n}= m_1] ={n\choose m_{1}} \frac{\alpha (n-m_1) +\theta m_1}{n(\theta+n-m_{1})_{(m_{1})}}(1-\alpha)_{(m_1-1)} .
\end{equation}
Within the population genetic setting of \cite{Don(86)}, Equation \eqref{eq_old} with $\alpha=0$ provides the distribution of the frequency of the oldest allele \citep{Kel(77),Wat(77)}.

Now, we can state our main results on the posterior distribution of the first $r$ order frequencies. Let $((T_{1},X_{1}),\ldots,(T_{n},X_{n}))$ be a random sample from the ordered PYP $Q$, and let $((T_{n+1},X_{n+1}),\ldots,(T_{n+m},X_{n+m}))$ be an additional random sample from the same ordered PYP $Q$. Moreover, we denote by $K_m^{(n)}$ the number of distinct species in the sample $((T_{n+1},X_{n+1}),\ldots,(T_{n+m},X_{n+m}))$ that are not in  $((T_{1},X_{1}),\ldots,(T_{n},X_{n}))$, i.e. $K_{m}^{(n)}=K_{n+m}-K_{n}$, and we denote by $W_{i, n+m}$ the frequency of the specie's label of order $i$ in the enlarged sample $((T_{1},X_{1}),\ldots,(T_{n+m},X_{n+m}))$, for $i = 1, \ldots, K_{n+m}$. To determine the distribution of the ordered frequencies $\mathbf{W}_{n+m}=(W_{1, n+m},\ldots , W_{K_{n+m}, n+m})$, it is useful to set
\begin{equation}\label{def_a}
A_{r}= \{ \text{species' labels with order }  1, \ldots , r \text{ are new} \},
\end{equation}
i.e., the event that the species' labels with higher weights have not been recorded in the first sample, and 
\begin{equation}\label{def_b}
B_{r}= \{ \text{species' labels with order }  1, \ldots , r \text{ are old} \},
\end{equation}
i.e. the event that the observations with higher weights have been recorded in the initial sample. 

\begin{theorem} \label{prp:posterior}
Let $((T_{1},X_{1}),\ldots,(T_{n},X_{n}))$ be a random sample under the BNP model \eqref{ex_model_ord}, such that the sample features $K_{n}=k$ distinct species with ordered frequencies $\mathbf{M}_{n}=\mathbf{m}$. Let $((T_{n+1},X_{n+1}),\ldots,(T_{n+m},X_{n+m}))$ be an additional random sample under the same BNP model \eqref{ex_model_ord} such that the enlarged sample $((T_{1},X_{1}),\ldots,(T_{n+m},X_{n+m}))$ features $K_{m+n}$ distinct species with corresponding ordered frequencies $\mathbf{W}_{n+m}$, and set $K_{m}^{(n)}=K_{n+m}-K_{n}$. If $A_{r}$ and $B_{r}$ are the events defined in \eqref{def_a} and \eqref{def_b}, respectively, then it holds:
\begin{itemize}
\item[i)] for $r \in \{ 1, \ldots , n+m\}$ such that $r \leq |\mathbf{w}|_{1:r} \leq m$,
\begin{align} \label{eq:dist_a_posteriori}
& \text{Pr} [A_r, W_{1,n+m} = w_1 , \ldots , W_{r,n+m} = w_r ,  K_m^{(n)} \geq r\, |\, K_n=k , \mathbf{M}_{n}=\mathbf{m}]\\
&\notag \quad ={m\choose w_{1},\ldots,w_{r},m-|\mathbf{w}|_{1:r}}\frac{C_{r,n+m}(\alpha,\theta,\mathbf{w})}{(\theta+n+m-|\mathbf{w}|_{1:r})_{(|\mathbf{w}|_{1:r})}};
\end{align}
\item[ii)] for $r \in \{  1, \ldots , k\}$ such that $0 \leq |\mathbf{w}|_{1:r} \leq m$,
\begin{align} \label{eq:dist_a_posteriori_old}
& \text{Pr}[B_r, W_{1,n+m} = w_1+m_1 , \ldots , W_{r,n+m} = w_r+m_r  \,|\, K_n=k , \mathbf{M}_{n}=\mathbf{m}]\\
&\notag\quad\times{m\choose w_{1},\ldots,w_{r},m-|\mathbf{w}|_{1:r}}\frac{\frac{C_{r,n+m}(\alpha,\theta,\mathbf{w}+\mathbf{m})}{(\theta+n+m-|\mathbf{w}+\mathbf{m}|_{1:r})_{(|\mathbf{w}+\mathbf{m}|_{1:r})}}}{\frac{C_{r,n}(\alpha,\theta,\mathbf{m})}{(\theta+n-|\mathbf{m}|_{1:r})_{(|\mathbf{m}|_{1:r})}}}
\end{align}
\end{itemize}
\end{theorem}

See Section~S1.5 of the Supplementary Material for the proof of Theorem \ref{prp:posterior}. Theorem \ref{prp:posterior} may be viewed as the posterior counterpart of Equation \eqref{eq:dist_a_priori}, with respect to an initial observable sample $((T_{1},X_{1}),\ldots,(T_{n},X_{n}))$. In particular, Equation \eqref{eq:dist_a_posteriori} and Equation \eqref{eq:dist_a_posteriori_old} provide two posterior distributions of the first $r$ order frequencies under the events $A_{r}$ and $B_{r}$, respectively, for $r\geq1$. Equation \eqref{eq:dist_a_posteriori} provides the posterior distribution of the first $r$ order frequencies having species' labels not belonging to the additional observable samples; that is the ordering of species' labels in the initial sample $((T_{1},X_{1}),\ldots,(T_{n},X_{n}))$ is changed according to the additional sample $((T_{n+1},X_{n+1}),\ldots,(T_{n+m},X_{n+m}))$. Equation \eqref{eq:dist_a_posteriori_old} provides the posterior distribution of the first $r$ order frequencies having species' labels belonging to the additional observable samples; that is the ordering of species' labels in the initial sample $((T_{1},X_{1}),\ldots,(T_{n},X_{n}))$ is not changed according to the additional sample $((T_{n+1},X_{n+1}),\ldots,(T_{n+m},X_{n+m}))$. BNP estimators of the first $r$ order frequencies, with respect to a squared loss function, are obtained in terms of posterior expectations, i.e. the vectors of expected values with respect to the posterior distributions \eqref{eq:dist_a_posteriori} and \eqref{eq:dist_a_posteriori_old}. As a corollary of Theorem \ref{prp:posterior}, we obtain the posterior distributions of the frequency of order $1$.

\begin{corollary} \label{cor:oldest_species_posterior}
Let $((T_{1},X_{1}),\ldots,(T_{n},X_{n}))$ be a random sample under the BNP model \eqref{ex_model_ord}, such that the sample features $K_{n}=k$ distinct species with ordered frequencies $\mathbf{M}_{n}=\mathbf{m}$. Let $((T_{n+1},X_{n+1}),\ldots,(T_{n+m},X_{n+m}))$ be an additional random sample under the same BNP model \eqref{ex_model_ord} such that the enlarged sample $((T_{1},X_{1}),\ldots,(T_{n+m},X_{n+m}))$ features $K_{m+n}$ distinct species with corresponding ordered frequencies $\mathbf{W}_{n+m}$, and set $K_{m}^{(n)}=K_{n+m}-K_{n}$. If $A_{1}$ and $B_{1}$ are the events defined in \eqref{def_a} and \eqref{def_b}, respectively, then it holds:
\begin{itemize}
\item[i)]
\begin{align} \label{eq_old_posterior}
&\text{Pr}[A_1, W_{1,n+m}= w_1 , K_m^{(n)} \geq 1\,|\, K_n=k , \mathbf{M}_{n}=\mathbf{m}] \\
& \notag\quad = \binom{m}{w_1}\frac{\alpha (n+m-w_1) +\theta w_1 }{(n+m)(\theta+n+m-w_{1})_{(w_{1})}}(1-\alpha)_{(w_1-1)}; 
\end{align}
\item[ii)]
\begin{align} \label{eq_posterior_r1_old_simplified}
&\text{Pr}[B_1, W_{1,n+m}= w_1+m_1 \,|\, K_n=k , \mathbf{M}_{n}=\mathbf{m}]\\
&\notag\quad = \binom{m}{w_1}\frac{\frac{\alpha (n+m-w_1-m_1) +\theta (w_1+m_1)}{(n+m)(\theta+n+m-w_{1}-m_{1})_{w_{1}+m_{1}}}}{\frac{\alpha(n-m_1)+\theta m_1}{n(\theta+n-m_{1})_{(m_{1})}}}(m_1-\alpha)_{(w_1)};
\end{align}
\item[iii)]
\begin{align} \label{eq:W1_distribution}
&\text{Pr}[ W_{1,n+m}= w \,|\, K_n=k ,\mathbf{M}_{n}=\mathbf{m}]\\
& \notag\quad = \frac{[\alpha (n+m-w)+\theta w](1-\alpha)_{(w-1)}}{(n+m)(\theta+n+m-w)_{(w)}}\\
& \notag\quad\quad  \times \Big[ \indic_{\{ 1, \ldots , m\}} (w) \binom{m}{w}  + \indic_{\{ m_1, \ldots , m_1+m \}}(w)  \binom{m}{w-m_1}  \frac{n(\theta+n-m_{1})_{(m_{1})}}{[\alpha (n-m_1)+\theta m_1](1-\alpha)_{(m_1-1)}} \Big].
\end{align}
\end{itemize}
\end{corollary}

The proofs of Equation \eqref{eq_old_posterior} and Equation \eqref{eq_posterior_r1_old_simplified} follow directly from Theorem \ref{prp:posterior} by setting $r=1$, whereas Equation \eqref{eq:W1_distribution} follows by combining \eqref{eq_old_posterior} and \eqref{eq_posterior_r1_old_simplified}. Within the population genetic setting of \cite{Don(86)}, Equation \eqref{eq:W1_distribution} with $\alpha=0$ provides the posterior distribution of the frequency of the oldest alleles. By exploiting Corollary \ref{cor:oldest_species_posterior}, BNP estimators of the frequency of order $1$, with respect to a squared loss function, are obtained in terms of the expected values of the posterior distributions \eqref{eq_old_posterior}, \eqref{eq_posterior_r1_old_simplified} and \eqref{eq:W1_distribution}. Here, we report the BNP estimator with respect to the posterior distribution \eqref{eq:W1_distribution} and we refer to Section~S1.6 of the Supplementary Material for the BNP estimators with respect to the posterior distributions  \eqref{eq_old_posterior} and \eqref{eq_posterior_r1_old_simplified}. In particular, if we set 
\begin{align*}
& C(\alpha, \theta,n,m,m_1)\\
&\quad=  [\alpha (n+m-m_1)+\theta m_1]\left[m_1 +m\frac{m_1 -\alpha}{\theta+n-\alpha}\right]\\
& \quad\quad +\left[m(\theta-\alpha)\frac{m_{1}-\alpha}{\theta+n-\alpha}\right]\left[(m_{1}+1)+(m-1)\frac{m_{1}+1-\alpha}{\theta+n+1-\alpha}\right]
\end{align*}
then
\begin{align} \label{eq:posteriorE}
&\E [W_{1,n+m}\,|\,  K_n =k ,\mathbf{M}_{n}=\mathbf{m}]\\
&\notag\quad=\frac{m(\theta+n+1-\alpha)_{(m)}}{(n+m)(\theta+n)_{(m)}} \frac{\theta+n\alpha}{\theta+n+1-\alpha}+ \frac{n(\theta+n-\alpha)_{(m)}}{(n+m)(\theta+n)_{(m)}} \frac{C(\alpha, \theta,n,m,m_1)}{\alpha(n-m_1)+\theta m_1}.
\end{align}
We refer to Section~S1.7 of the Supplementary Material for explicit expressions of the posterior probabilities of the events $A_1$ and $B_1$, which complete the main result of Corollary~\ref{cor:oldest_species_posterior}.

\subsection{Estimation of prior's parameter $(\alpha,\theta)$} \label{sec:prior_estimation}

The closed-form expressions of our results facilitate posterior inferences. In particular, if the prior's parameters $(\alpha,\theta)$ are estimated with an empirical Bayes approach and fixed, then the inferential procedure becomes straightforward and efficient. Here we consider both an empirical Bayes approach and a fully Bayes approach for the estimation of $(\alpha,\theta)$, the latter considering the specification of a prior distribution on $(\alpha,\theta)$. While the empirical Bayes approach takes advantage of closed-form formulae, the fully Bayes approach may sometimes be preferable. In both cases, the inference is based on an initial sample of $n$ observations, which are then used to make predictions on a second set of $m$ data points.

Within the empirical Bayes approach, we consider methods relying on maximum likelihood estimation and methods relying on moment-based estimation. With regards to maximum likelihood estimation, the problem consists in finding the values of $\alpha$ and $\theta$ that maximize the (marginal) likelihood function, which in this context is equal to the EPPF.
Under the ordered PYP, this coincides with \eqref{eq_eppf_ord}, and the parameters found by solving:
\begin{displaymath}
\max_{\alpha,\theta} \Phi_k^{(n)}(m_1, \ldots, m_K; \alpha,\theta).
\end{displaymath}
As a term of comparison, we also consider the performance of estimating the prior's parameters when the model is misspecified, specifically assuming an ordered DP prior (i.e. fixing $\alpha = 0$), or ignoring the ordering structure or the model, i.e. maximizing the EPPF of the standard PYP \eqref{eq_eppf}. The sets of prior's parameters obtained by optimizing these EPPFs are respectively denoted with \texttt{ordPYP}, \texttt{ordDP} and \texttt{stdPYP}. Note that the approaches based on the ``misspecified'' likelihood (\texttt{ordDP} and \texttt{stdPYP}) do not take full advantage of the increased complexity of the model, and are considered only as reference. When the model is correctly specified, these methods do not estimate the correct parameters.

With regard to moment-based estimation, we consider a statistic of interest and then match its population first moment with the corresponding observed sample statistic. Rather than doing this for the full initial sample, we consider a collection of samples of increasing size and match the statistics of interest's trajectory given by increasing sample sizes, using a least squares method. Specifically, we consider a grid of sample size values $1 \leq n_1 \leq \ldots \leq n_d = n$; for each $n_i$, we compute the discrepancy between the first moment and the observed statistics for the first $n_i$ observations; finally, we minimize the sum of the squared discrepancies, where the sum is for $i = 1, \ldots, d$. 
The statistics we consider are the frequency of the first ordered species $M_{1,n}$ and the number of distinct species $K_n$; the parameters obtained are respectively denoted as \texttt{lsM1} and \texttt{lsK}:
\begin{align*}
(\hat\alpha,\hat\theta)_{\texttt{lsM1}}=&\argmin_{\alpha,\theta} \sum_{i=1}^d \left( \E[M_{1,n_i};\alpha,\theta] - M_{1,n_i} \right)^2\\
(\hat\alpha,\hat\theta)_{\texttt{lsK}}=&\argmin_{\alpha,\theta} \sum_{i=1}^d \left( \E[K_{n_i};\alpha,\theta] - K_{n_i} \right)^2,
\end{align*}
where $M_{1,n_i}$ and $K_{n_i}$ are the frequency of the first ordered species and the number of distinct species in the first $n_i$ samples, respectively. In Section~S2.1 of the Supplementary Material, we provide additional details, as well as report the pseudocode algorithms, to obtain these parameter estimates. Moment-based estimation allows to focus the estimation problem on a  specific feature or property of the data, by choosing the summary statistics of interest. Moreover, by considering a grid of sample size values, the method learns the growth curve over $n$, ideally being more robust compared to methods that only look at one ``snapshot'' given by the full dataset. Additionally, because they do not rely on the full likelihood, these methods could be more robust in the case of model misspecification.

We also consider a fully Bayes approach (\texttt{FB}), by specifying a suitable prior distribution for the parameters $(\alpha, \theta)$ and focusing the inference on the posterior distribution. This can be implemented using standard MCMC algorithms. Here, we consider independent non-informative prior distributions, setting $p(\theta) = G(0.1,0.1)$ and $p(\alpha) = Unif(0,1)$.


\section{Numerical illustrations}\label{sec4}

We empirically study the ordered PYP and assess the performance of our BNP approach for estimating the frequency of the first ordered species, using synthetic and real data. Moreover, we empirically study distributional properties of the ordered random partitions induced by the model, in particular focusing on the species' ordering distribution. In the synthetic data, we compare the performance when the data is generated from the model, i.e. the PYP prior is correctly specified, and when the data is generated from different distributions, thus under model misspecification. For the application of our model to real data, we analyze genetic variation using samples from the 1000 Genome Project \citep{10002015global}, which we combine with variants' age estimates obtained from the Human Genome Dating Project \citep{Alb(20)}.  Studying genetic variation is of great importance to investigate population structure, to detect related samples, to investigate demographic history, and to learn about the risk of diseases and different quantitative traits. By also incorporating information on the variants' age, similar and further issues can be assessed, such as using the variants' age distribution to compare populations, differentiating age distributions in pathogenic and benign variants, and learning about genealogical history \citep{Alb(20)}. 
Here we assess the performance of our method on predicting frequencies of variants ordered by their age, focusing in particular on the oldest variant. 
Code is available at \url{https://github.com/cecilia-balocchi/OrderedSSP}.

\subsection{Preliminaries}

While the distributional properties of the random partition induced by the PYP are well-known in the BNP literature, the properties induced by the ordered PYP are less understood. Because of the marginality property \eqref{eq:interpl}, some distributional properties of the ordered PYP are equivalent to the ones of the PYP. For example, the distribution of the number of distinct species (or clusters) induced by the ordered PYP is equal to the one induced by the PYP, reported in equation (S3).  However, other properties that relate to the ordering on the species are not as well understood. In particular, we are interested in learning the behavior of the distribution that assigns the order to a new cluster. We aim to characterize it using simple descriptive features. 
We achieve this goal by studying the predictive distribution \eqref{eq:prednew} that 
assigns the $n+1$ observation to a \textit{new} species of order $j$, and to examine its behavior marginalizing on all the configurations of partitions of $n$ observations into $K_n$ species.  In other words, we empirically study the marginal distribution $ \text{Pr}(\text{order}(n+1) = j \vert (T_{n+1},X_{n+1})=(T^*_{\text{new}},X^*_{\text{new}}), K_n = k)$ that a new species after $n$ observations is assigned order $j$, given that the previous $n$ observations are partitioned into $K_n=k$ species, for $j= 1, \ldots, K_n+1$.  This distribution is ``marginal'' compared to \eqref{eq:prednew}, because it does not condition on the frequencies of the ordered species, $\mathbf{M}_{n}=(m_1,\ldots,m_{k})$, and it is obtained by marginalizing over all possible configurations $(m_1,\ldots,m_{k})$ of partitions of $n$ into $K_n=k$ species. 

In Figure \ref{fig:ordering_distr} we depict the marginal ordering distribution for a new species given an observed sample of size $n = 10$. The solid color lines represent the ordering distributions given the number of previously observed clusters, marginally on the partition configuration (different colors correspond to different numbers of clusters), $ \text{Pr}(\text{order}(n+1) = j \vert (T_{n+1},X_{n+1})=(T^*_{\text{new}},X^*_{\text{new}}), K_n = k)$. The colored points instead represent the realizations of the ordering probabilities conditional on individual partition configurations $\mathbf{M}_{n}$, $ \text{Pr}(\text{order}(n+1) = j \vert (T_{n+1},X_{n+1})=(T^*_{\text{new}},X^*_{\text{new}}), K_n = k, \mathbf{M}_{n}=(m_1,\ldots,m_{k}))$, for different values of $(m_1,\ldots,m_{k})$; note that this conditional ordering probability can be found from \eqref{eq:prednew} as $q^{\tiny{(new)}}_{j} / \sum_{i=1}^{k+1} q^{\tiny{(new)}}_{i}$.
The ordering distribution has been sketched for different configurations (parameters) of the PYP prior. In particular, from left to right, we represent the distribution under the ordered DP ($\theta > 0, \alpha = 0$, first panel), the ordered PYP with $\alpha < \theta$, $\alpha = \theta$, and $\alpha > \theta$ and the ordered $\alpha$-stable process ($\theta = 0, \alpha > 0$, last panel). Figure \ref{fig:ordering_distr} shows that the ordering distribution changes depending on the parameters $\theta$ and $\alpha$: for $\theta > \alpha$ the probability that a new cluster is assigned to order $j$ increases with $j$ for each previous number of species $K_n$ (first and second panels from the left) and the increasing trend in stronger when the difference $\theta-\alpha$ is large. For $\theta = \alpha$ (third panel from the left) the trend is constant over $j$, for all $K_n$. For $\theta < \alpha$ we see instead that the trend is decreasing with $j$, for each $K_n$ (fourth panel). The last panel shows that in the case of the $\alpha$-stable process ($\theta = 0 < \alpha$) the trend is again constant. These intuitions are useful for constructing a distribution for ordered partitions that has properties similar to those induced by the ordered PYP prior. Moreover, Figure \ref{fig:ordering_distr} emphasizes an additional aspect of the ordered PYP's improved flexibility compared to the ordered DP ($\alpha = 0$).

\begin{figure}[t]
\centering
\includegraphics[width=0.97\linewidth]{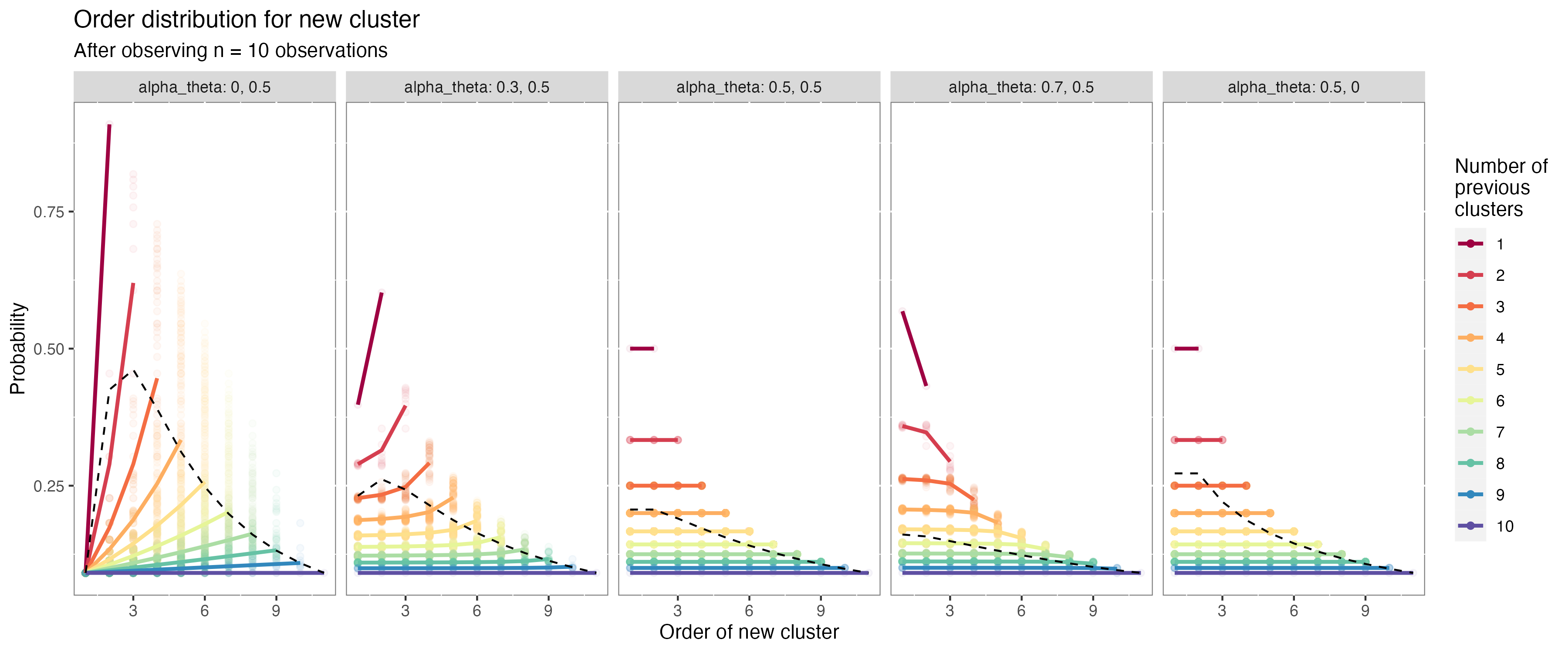} 
\caption{Order distribution for a new species, given a sample of $n = 10$ observations, divided into $K_n = 1, \ldots, 10$ species (each color represents a different value of $K_n$). Solid colored curves represent the ordering distributions marginally on the partition configuration $\mathbf{M}_n$,  $ \text{Pr}(\text{order}(n+1) = j \vert (T_{n+1},X_{n+1})=(T^*_{\text{new}},X^*_{\text{new}}), K_n = k)$; the colored points instead show the variation of the ordering distribution conditional on $\mathbf{M}_n$, across different partition configurations $\mathbf{M}_n=(m_1,\ldots,m_{k})$. The dashed black line represents the order distribution for a new species, marginally on the number of previous species $K_n$.} \label{fig:ordering_distr}
\end{figure}

\subsection{Analysis of synthetic data}

We consider the problem of making inference on some quantities of interest with our BNP approach, in the context of synthetic data.  In particular, we study the performance on an additional sample of size $m$ of the posterior mean predictors for the total number of species $K_{n+m}$, the frequency of the oldest cluster $W_{1,n+m}$, and the frequency $W_{1,n+m}$ conditionally on the knowledge that either the event $A_1$ or $B_1$ happened ($W_{1,n+m}\vert A_1$ and $W_{1,n+m}\vert B_1$). For notational simplicity, we will remove the dependence on $n+m$ in the notation. All of the estimators for these quantities have closed-form espressions, thanks to the results in Section Section~\ref{sec3}. In particular, we estimate $K_{n+m}$ using the posterior mean of the number of unseen species $K_{m}^{(n)}$ (see Section S1.3 of the Supplementary Materials), while $W_{1,n+m}$ is estimated using \eqref{eq:posteriorE}, and $W_{1,n+m}\vert A_1$ and $W_{1,n+m}\vert B_1$ using respectively combining formulas (S14) with (S20), and (S19) with (S21) from the Supplementary Materials.
We compare the predictive performance of the ordered PYP model, under a full Bayes approach, and when the prior's parameters are estimated with the empirical Bayes methods described in Section~\ref{sec:prior_estimation}. We generate the synthetic data under different scenarios. We first consider a framework where the model is correctly specified, i.e. the data is generated from the ordered PYP prior. We then focus on a framework where the model is misspecified, as the data is generated from different distributions. 

\subsubsection{Inference of the first ordered frequency under correct specification} \label{sec:mod_sim}

We first generate the data from the ordered PYP prior, i.e. under correct model specification. Specifically, we consider 100 datasets of size $n = 500$, generated from the ordered PYP prior with randomly sampled parameters $(\theta, \alpha)$; for each dataset, we consider 25 additional datasets of size $m = 5000$, and compute the median prediction error across the 25 additional datasets (absolute percentage error is computed for the four quantities of interest $K$, $W_1$, $W_1 \vert A_1$, $W_1 \vert B_1$). 

Figure~\ref{fig:model_sim} compares the performance of the full Bayes approach with the different empirical Bayes approaches, for the four quantities of interests.
Overall we notice that the performance of the full Bayes approach (\texttt{FB}) is almost identical to the one of the approach based on the EPPF of the ordered PYP (\texttt{ordPYP}). 
These two approaches tend to have the best prediction error for all the quantities of interest, except for the estimation of the frequency of the first ordered species $W_1$, where the best predictive performance is achieved by \texttt{lsM1}. 
The empirical Bayes approach based on the standard PYP likelihood \texttt{stdPYP} has a similar but slightly worse performance to \texttt{ordPYP} and \texttt{FB}. \texttt{ordDP} instead shows poor predictive performance for all quantities of interests.

\begin{figure}[t]
\centering
\includegraphics[width=0.97\linewidth]{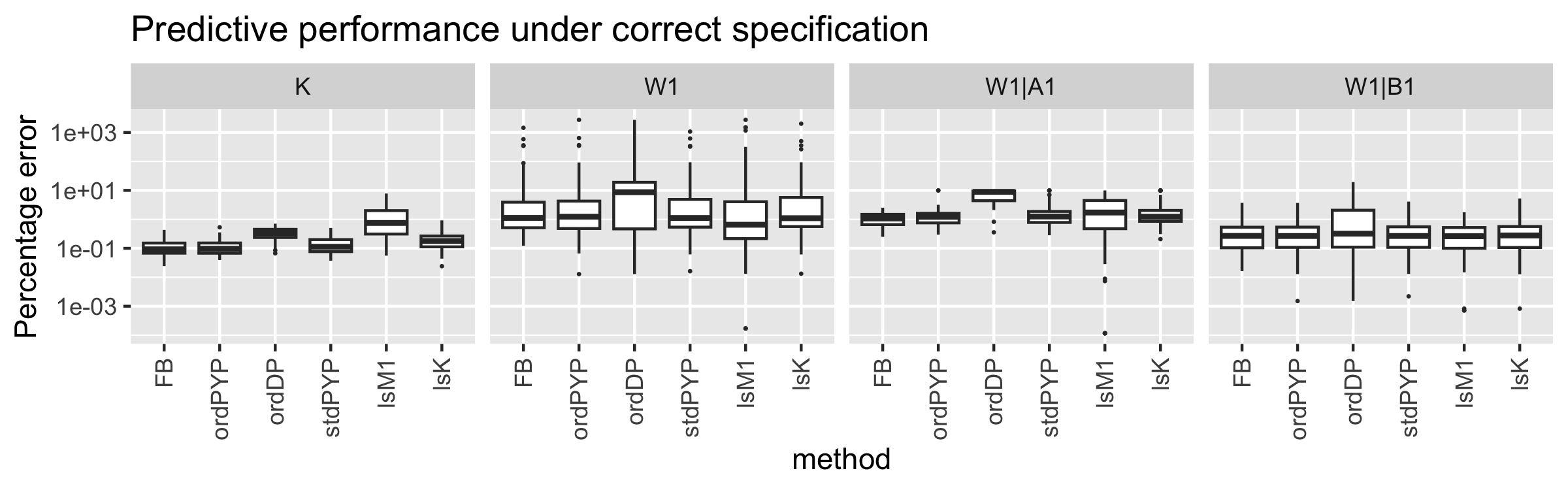} 
\caption{Predictive performance for our BNP approach for the total number of distinct species $K$, the frequency of the first ordered cluster $W_1$, and the conditional frequencies $W_1\vert A_1$ and $W_1\vert B_1$ (shown in different columns). The boxplots display the median absolute percentage error between the predicted posterior mean and the true value, across several datasets simulated from the ordered PYP model. The full Bayes approach and different parameter-estimating methods are compared.} \label{fig:model_sim}
\end{figure}

\subsubsection{Inference of the first ordered frequency under misspecification}

We then consider a synthetic data framework where the data is not generated from the ordered PYP prior, i.e. under model misspecification. We aim at evaluating how our method performs in such adverse conditions, which are in general also to be expected from real data. We consider several different data-generating processes. In all of them, we generate the species order independently from the observations' species (or cluster) assignment. For the clustering distribution, we consider: (a) the standard \texttt{DP}, (b) the standard \texttt{PYP}, and (c) the (infinite support) \texttt{Zipf} distribution (also known as the zeta distribution). In the latter, each observation is associated with an integer (sampled from the Zipf distribution) and clusters are formed by aggregating observations mapped to the same integer. Given that these distributions do not induce an order on the species, we additionally consider an ordering distribution for each new species. For the ordering distribution, we consider: (1) the \texttt{alpha-stable} distribution (induced by the ordered $\alpha$-stable process, a special case of the ordered PYP with $\theta = 0$ where the predictive distribution \eqref{eq:prednew} simplifies significantly), and (2) what we call the \texttt{arrival-weighted} distribution. The latter considers the ordering induced by the cluster arrival and introduces a random component by sampling for each new cluster an exponentially distributed weight with mean given by the arrival order, and ordering the species according to the weight.
We generate 100 datasets of size $n = 500$ using randomly sampled parameters, and for each of them consider 25 additional datasets of size $m = 5000$; we consider the median absolute percentage errors across the 25 additional datasets.

Figure~\ref{fig:misspec_sim} displays the predictive performance results for the different measures of interests ($K$, $W_1$, $W_1\vert A_1$ and $W_1\vert B_1$) across different rows, and for different data-generating distributions (clustering and ordering distributions) across different columns. We compare the performance of the full Bayes and of the different empirical Bayes approaches.
The first row of Figure~\ref{fig:misspec_sim} focuses on the prediction of the number of distinct species $K$. Overall, the best prediction error is achieved by the empirical Bayes method maximizing the likelihood of the non-ordered (standard) PYP (\texttt{stdPYP}). This is not surprising given that the ordered PYP is not correctly specified, and the number of distinct species has the same behavior under the ordered and the standard PYP. The second best performance is achieved by \texttt{lsK}, and this is consistent with the fact that it was designed to be more robust and learn this ``feature'' ($K$) even under model misspecification. The performance of the full Bayes approach \texttt{FB} and the empirical Bayes approach based on the ordered PYP likelihood \texttt{ordPYP} are similar and slightly worse when the clustering distribution is the PYP or the Zipf distribution, but it gets considerably worse when the clustering distribution is the DP, suggesting that these methods are not able to learn that $\alpha$ is equal to zero in such simulated datasets. The performance of \texttt{ordDP} is good under the DP clustering distribution, but quite bad otherwise.

The second row of figure~\ref{fig:misspec_sim} shows the error committed when predicting the frequency of the first ordered species $W_1$. The best performance is often achieved by \texttt{lsM1}, followed by \texttt{ordPYP} and \texttt{FB}. When the clustering distribution is the Zipf distribution, the difference between \texttt{lsM1} and the other two is considerable, with the former being more robust. Sometimes, \texttt{stdPYP} achieves comparable results. 
Overall, the distribution of the errors is much more spread out when the ordering is generated from the \texttt{arrival-weighted} distribution, compared to the \texttt{alpha-stable} distribution, meaning that the behavior induced by the former is more different from the one described by \ref{eq:posteriorE}.
Finally, the third and the fourth rows focus on the predictive performance for the conditional frequencies $W_1\vert A_1$ and $W_1\vert B_1$. 
In both cases, the best performance is achieved by \texttt{ordPYP}, \texttt{FB}, and \texttt{lsM1}. 

\begin{figure}[t!]
\centering
\includegraphics[width=0.97\linewidth]{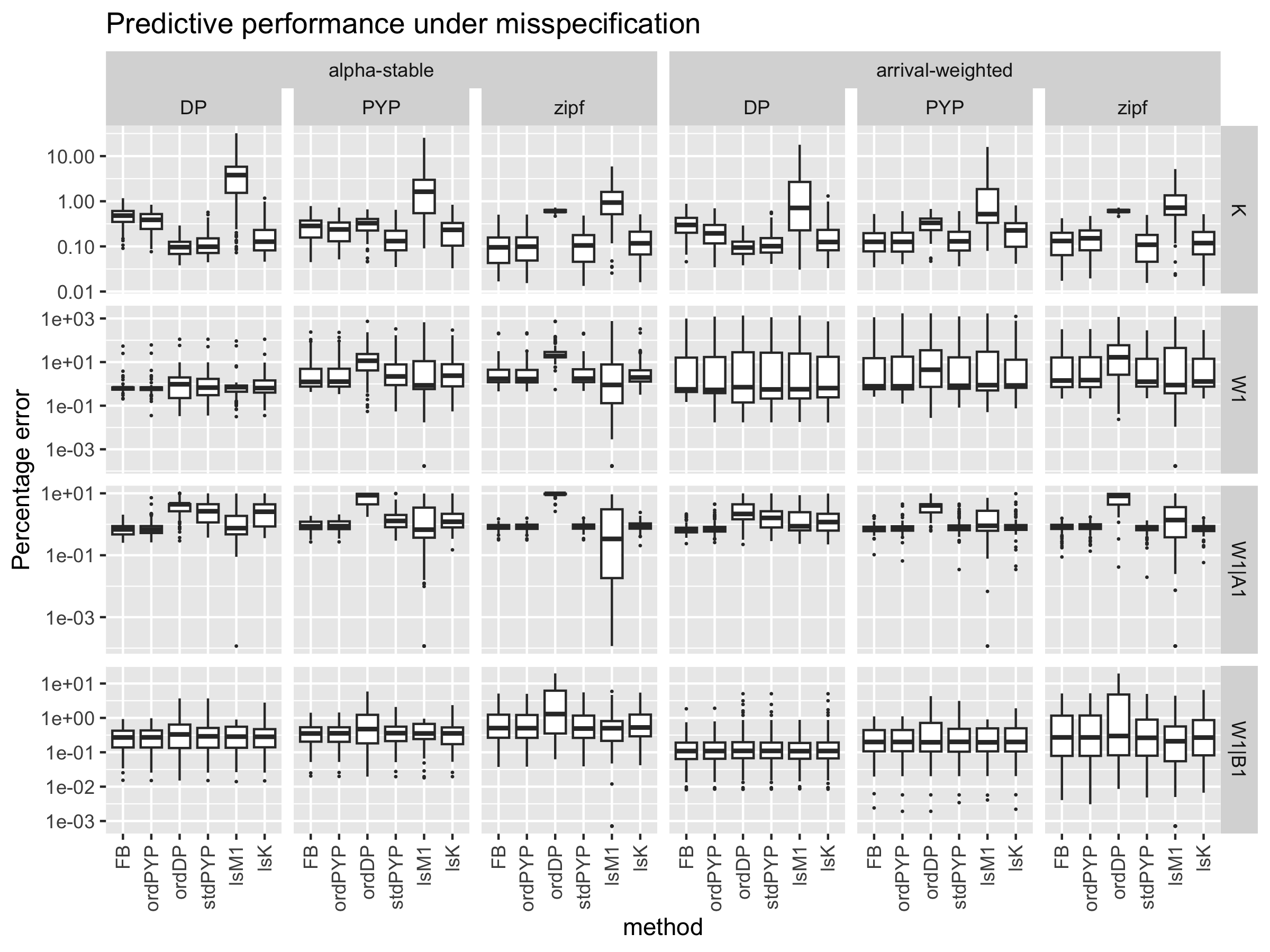} 
\caption{Predictive performance for our BNP approach, for the total number of clusters $K$, the frequency of the first ordered cluster $W_1$, and the conditional frequencies $W_1\vert A_1$ and $W_1\vert B_1$ (shown in different rows).
The boxplots display the median absolute percentage error between the predicted posterior mean and the true value, across several datasets simulated under different distributions (shown in different columns). We compare various parameter-estimating methods. 
} \label{fig:misspec_sim}
\end{figure}

Overall, the full Bayes approach achieves results very similar to the ones of the empirical Bayes approach based on the likelihood of the ordered PYP (\texttt{ordPYP}).
When the model is misspecified, \texttt{lsM1} seems to be more robust for the estimation of both the marginal distribution and the conditional distribution of $W_1$, but it performs quite poorly for the prediction of $K$. While the parameter-estimation method \texttt{ordPYP} and the full Bayes approach \texttt{FB} do not always produce the best predictions, they seem to provide a good balance between learning the number of species $K$ and the frequency of the first order species $W_1$. Alternatively, the parameter estimation method could be selected based on the quantity of interest. 

\subsection{Analysis of genetic variation data} \label{sec:genetic}

We employ our method to analyze genetic variation data from the 1000 Genomes Project \citep{10002015global}. We examine single nucleotide polymorphisms (SNPs) corresponding to certain genes for a sample of 2548 individuals. In particular, we consider a variant to be present for a given individual and SNP locus if the present DNA base is different from the reference genome in either allele. We combine the variants from all individuals and all SNPs corresponding to a certain gene, to create our basic sample, where the variant location (i.e. the SNP) represents the species each sample unit belongs to. To associate an ordering to each species, we used SNPs age estimates from Human Genome Dating project \citep{Alb(20)}, and discarded any variant for which no age information is available. We study two datasets, formed by the genetic variants corresponding respectively to the BRCA2 and the EDAR genes, for which we analyze respectively 1482 and 1073 unique SNPs. We focus on the predictive performance for the number of distinct variants $K$ and the frequency of the oldest variant $W_1$, using the different parameter estimation methods for our BNP approach.
We repeat our analysis for 100 different training and testing sets, randomly sampled so that the testing set size is approximately 20 times larger than the training set size. 

\begin{figure}[t!]
\centering
\includegraphics[width=0.98\linewidth]{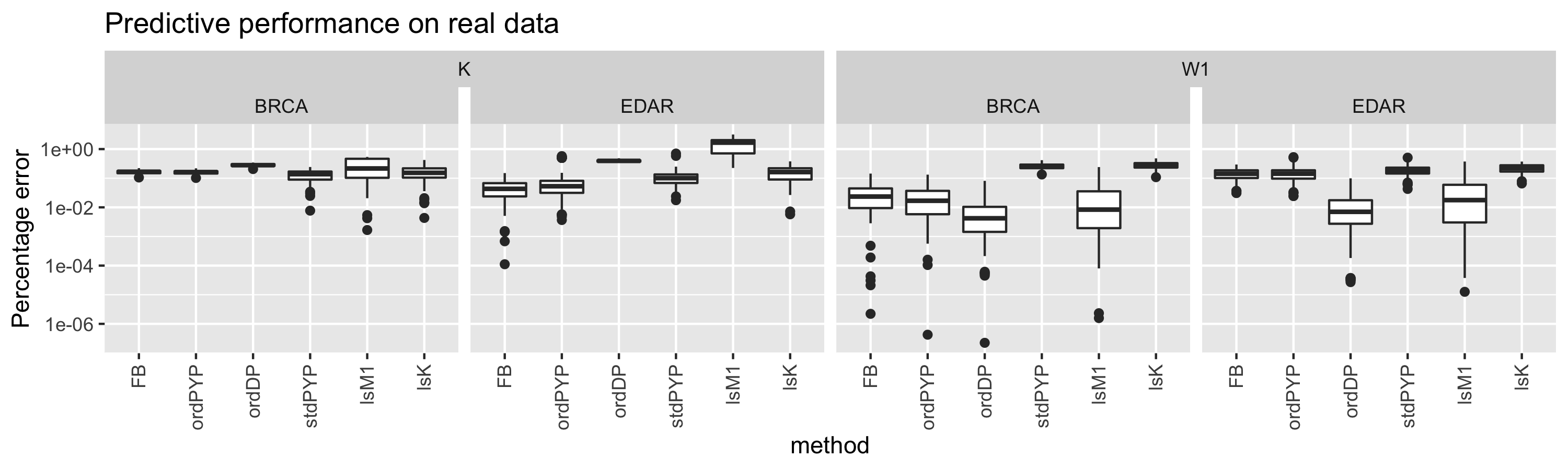}
\caption{Prediction performance for the number of distinct species $K$ and the frequency of the oldest cluster $W_1$ across several training-testing sets for BRCA2 (panels one and three) and EDAR (panels two and four).} \label{fig:crossval}
\end{figure}

In Figure~\ref{fig:crossval} we report the percentage errors for predicting $K$ and $W_1$ for the datasets corresponding to the BRCA2 and the EDAR genes. For the prediction of the number of species, we find the method \texttt{stdPYP} provides the best results for the BRCA2 dataset, but that \texttt{ordPYP} and \texttt{FB} have the best prediction for the EDAR dataset. In terms of predicting the frequency of the oldest variant, \texttt{ordDP} and \texttt{lsM1} perform better for both. We note that the performances of \texttt{ordPYP} and \texttt{FB} are comparable to that of \texttt{ordDP} in the BRCA2 dataset, as the former often estimated values of $\alpha$ close to zero; however, the two methods achieved different performance in the EDAR dataset, due to the fact that \texttt{ordPYP} and \texttt{FB} estimated positive values of $\alpha$. Under further inspection, we note that \texttt{lsM1} has performance comparable to \texttt{ordDP} in the prediction of $W_1$ for the EDAR dataset because it often estimates values of $\alpha$ very close to zero.
It is also surprising to see that these datasets seem to have power-law behavior in terms of the number of species, for which values of $\alpha$ greater than 0 provide more accurate predictions, but the best predictions for the frequency of the oldest variant $W_1$ are produced by methods that estimate $\alpha = 0$. This is probably due to the model being not correctly specified for these data. Thus we recommend analyzing the two prediction problems separately.

\begin{figure}[t!]
\centering
\includegraphics[width=0.97\linewidth]{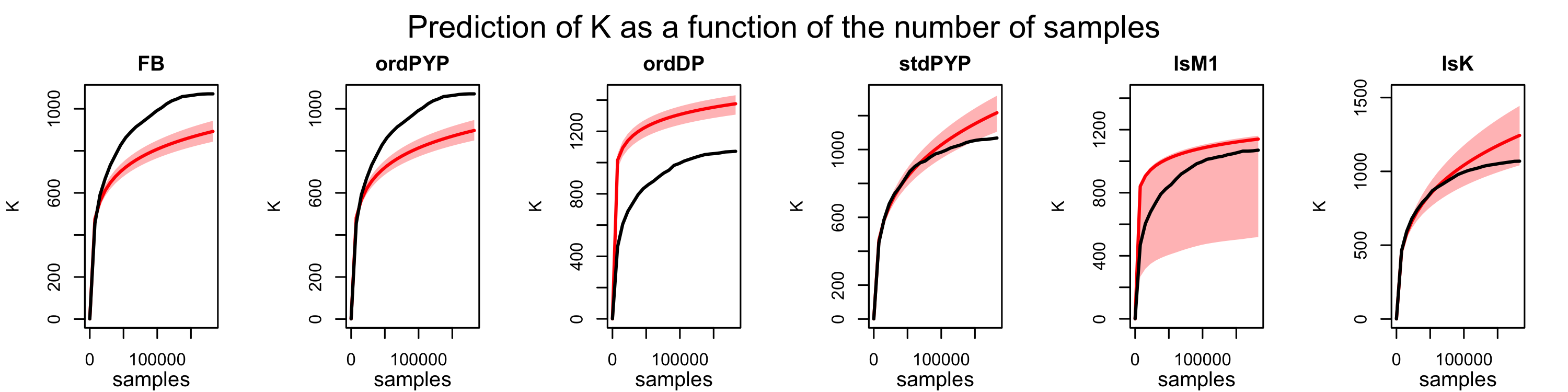} \\
\includegraphics[width=0.97\linewidth]{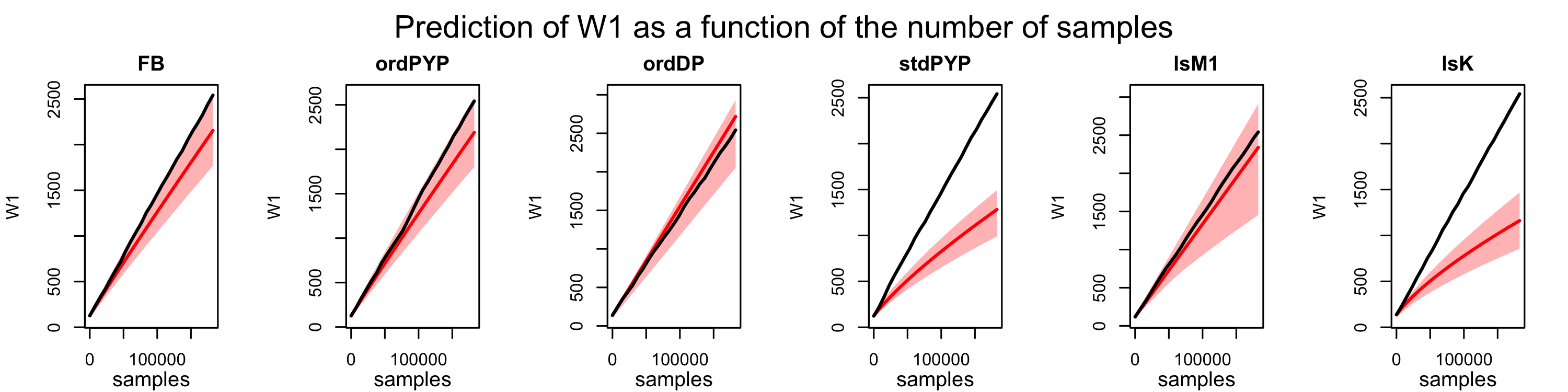} 
\caption{Prediction of the curve of $K$ (top panels) and $W_1$ (bottom panels) as a function of sample size, for the BRCA2 dataset. The curves represented correspond to the training-testing split that resulted in the median error, with the observed curve depicted in black, the predicted curve in red, and the red confidence bands corresponding to 95\% empirical quantiles.} \label{fig:median_curves}
\end{figure}

Our method can also be used to analyze the predictions as functions of the number of samples, rather than simply make a prediction for the entire test set. Using the multiple training and testing splits method, we also looked at the performance in the prediction of the curve of $K$ and $W_1$. Figure~\ref{fig:median_curves} reports the results for each parameter-estimating method, applied to the BRCA2 dataset. Similar plots are reported in Section~S2.2 of the Supplementary Material for the EDAR gene dataset.
For simplicity of visualization, we depict in black the actual curve of $K$ (or respectively $W_1$) observed in the training-testing split that most closely represents the median error achieved by each method. The prediction curve (depicted in red) is also the one corresponding to the ``median'' training-testing split, while the red bands represent the $95\%$ confidence bands. We report the empirical confidence bands, computed using the empirical quantiles of the curve estimates across the various training-testing splits.
For the estimation of the number of distinct species $K$, the top panel of Figure~\ref{fig:median_curves} confirms that \texttt{stdPYP} and \texttt{lsK} have the best performance, with quite accurate prediction and with the represented curve falling within the 95\% bands. In the other parameter-estimating methods the prediction is much worse, and often the curve does not fall within the confidence bands.
In terms of predictions for $W_1$, the results are consistent with the expectations from the predictions on the whole test set, with \texttt{ordDP}, \texttt{ordPYP}, \texttt{FB} and \texttt{lsM1} performing better, with very accurate prediction of $W_1$ as a function of the number of additional samples.


\section{Discussion}\label{sec5}

A common feature of SSPs is the invariance with respect to species labeling, i.e. species’ labels are immaterial in the definition of the functional of interest, which is at the core of the development of the BNP approach to SSPs under the popular PYP prior \citep{Lij(07),Lij(08),Fav(09),Fav(13)}. In this paper, we considered SSPs that are not ``invariant'' to species labeling, in the sense that an ordering or ranking is assigned to species’ labels, and we developed a BNP approach to such problems. In particular, inspired by the seminal work of \citet{Don(86)} on age-ordered alleles' compositions, with a renowned interest in the frequency of the oldest allele \citep{Cro(72),Kel(77),Wat(77)}, we studied the following SSP with ordering: given an observable sample from an unknown population of individuals belonging to some species (alleles), with species' labels being ordered according to some weights (ages), estimate the frequencies of the first $r$ order species’ labels in an enlarged sample obtained by including additional unobservable samples from the same population. Our BNP approach relies on the ordered PYP prior, which leads to an explicit posterior distribution of the first $r$ order frequencies, with corresponding estimates being simple and also computationally efficient. We presented an empirical validation of our approach on both synthetic and real data. The proposed methodology has been applied to the analysis of the genetic variation, showing its effectiveness in the estimation of the frequency of the oldest allele.

The sampling structure of the ordered PYP has been first presented in \citet{Gne(05)} and \citet{Jam(06)}, and since then no other works have further investigated such a sampling structure in BNPs. To date, the sampling formulae of ordered PYP prior appear to be largely unknown and unexplored in the BNP literature. Our work highlights the great potential of the ordered PYP in BNP inference for SSPs with ordering, paving the way for future research. First, in our work we considered the problem of estimating the frequency of the first $r$ order species, which is arguably the most natural SSP with ordering; other (discrete) functionals of the ordered species' composition of unobservable samples may be of interest, e.g. the number of unseen species with order less than $r$ that would be observed in additional samples, and the number of unseen species observed with a certain prevalence and order less than $r$ in the enlarged sample. Second, while SSPs with ordering have a natural motivation in population genetics, they may be of interest in different fields; in Section~S2.3 of the Supplementary Material, we present an application in the context of citations to academic articles, where each cited article is a species whose order is determined by the publication date; another application is in the context of online purchasing of items, with order being determined by a suitable measure of popularity assigned to items. When analyzing these applications in terms of SSMs, some simplifications need to be considered, which might impact certain aspects of the analysis. For example, our model cannot consider the case of different species having the same ranking (such as different papers published on the same date). However, we argue that it is still interesting to understand if the ordered PYP is an adequate model to describe the behaviors observed in such applied contexts.

We believe that interest in ordered SSMs is not limited to BNP inference for SSPs with ordering. Ordered SSMs, and in particular the ordered PYP prior, may be also applied to a setting in which species are not explicitly observed, but need to be inferred. This is the case of Bayesian mixture models \citep{Fru(19)}, where each observation is assumed to be assigned to a latent component (species' label), which characterizes some features of the distribution of the observations. In such a setting, ordered SSMs may be used to specify the distribution of the mixture's probabilities of the latent components, with such components being ordered according to some weights \citep{Jam(06a)}. Ordered SSMs admit a natural extension to the features sampling framework, which generalizes the species sampling framework by allowing each observation to belong to multiple species, now called features. Feature sampling problems first appeared in ecology for modeling the presence or absence of an animal in a trap, and their importance has grown dramatically in recent years driven by numerous applications in biological and physical sciences \citep{Mas(21),Mas(22)}. In such a context, the Beta process prior \citep{Bro(13)} is the most popular nonparametric prior for modeling the unknown feature's composition of the population.  The definition of an ordered version of the Beta process prior, and generalizations thereof, would be the starting point to introduce and investigate feature sampling problems with ordering.


\subsection*{Acknowledgments}

The authors are grateful to the Associate Editor and two anonymous Referees for their comments and corrections that allow them to improve remarkably the paper. The authors wish to thank Paul Jenkins for useful discussions on the use of age-ordered random partitions in population genetics. 
Federico Camerlenghi is a member of the \textit{Gruppo Nazionale per l'Analisi Matematica, la Probabilit\`a e le loro Applicazioni} (GNAMPA) of the \textit{Istituto Nazionale di Alta Matematica} (INdAM). 
Stefano Favaro is also affiliated to IMATI-CNR ``Enrico Magenes'' (Milan, Italy).

\subsection*{Funding}
The first and third authors gratefully acknowledge funding from the European Research Council (ERC), under the European Union's Horizon 2020 research and innovation programme, Grant agreement No. 817257.
The second and third authors gratefully acknowledge the support from the Italian Ministry of Education, University and Research (MIUR), ``Dipartimenti di Eccellenza" grant 2023-2027.
The second author was supported by the European Union – Next Generation EU funds, component M4C2, investment 1.1., PRIN-PNRR 2022 (P2022H5WZ9).
 


\newpage

\begin{center}
{\Large {\bf Supplementary Materials for }}

\bigskip

{\Large {\bf ``A Bayesian Nonparametric Approach to Species Sampling Problems with Ordering"}}

\bigskip

\end{center}

\setcounter{equation}{0}
\setcounter{figure}{0}
\setcounter{table}{0}
\setcounter{page}{1}
\setcounter{section}{0}
\setcounter{theorem}{0}
\makeatletter
\renewcommand{\theequation}{S\arabic{equation}}
\renewcommand{\thefigure}{S\arabic{figure}}
\renewcommand{\thetable}{S\arabic{table}}
\renewcommand{\bibnumfmt}[1]{[S#1]}
\renewcommand{\citenumfont}[1]{S#1}

\section{Proofs and details}

\subsection{Proof of Equation (9)}  \label{app:1}

The proof of Equation (9) is based on the following identity
\begin{equation} \label{eq:induction}
\sum_{\pi \in S_k}\theta^{-1} \prod_{j=1}^k \frac{\alpha (m_{\pi(j+1)}+ \ldots + m_{\pi (k)}) +\theta m_{\pi (j)}}{(m_{\pi(j)}+ \ldots + m_{\pi (k)})} = \prod_{i=1}^{k-1} (\theta+i \alpha).
\end{equation}
that we are going to show by induction.
When $k=1$ this is trivially true. Now assume that (inductive hypothesis)
Equation \eqref{eq:induction} is true with $k-1$ in place of $k$, thus one has
\begin{align*}
& \frac{1}{\theta}\sum_{\pi \in S_k} \prod_{j=1}^k \frac{\alpha (m_{\pi(j+1)}+ \ldots + m_{\pi (k)}) +\theta m_{\pi (j)}}{(m_{\pi(j)}+ \ldots + m_{\pi (k)})} \\
&\qquad = \frac{1}{\theta}\sum_{\ell=1}^k \sum_{\pi \in S_{k}^{(\ell)}} \prod_{j=1}^k\frac{\alpha (m_{\pi(j+1)}+ \ldots + m_{\pi (k)}) +\theta m_{\pi (j)}}{(m_{\pi(j)}+ \ldots + m_{\pi (k)})}
\end{align*}
where $S_k^{(\ell)}$ denotes all the permutations $\pi$ satisfying $\pi(1)=\ell$.
As a consequence one has
\begin{align*}
& \frac{1}{\theta}\sum_{\pi \in S_k} \prod_{j=1}^k \frac{\alpha (m_{\pi(j+1)}+ \ldots + m_{\pi (k)}) +\theta m_{\pi (j)}}{(m_{\pi(j)}+ \ldots + m_{\pi (k)})} \\
&\qquad = \frac{1}{\theta}\sum_{\ell=1}^k \sum_{\pi \in S_{k}^{(\ell)}} \prod_{j=2}^{k}\frac{\alpha (m_{\pi(j+1)}+ \ldots + m_{\pi (k)}) +\theta m_{\pi (j)}}{(m_{\pi(j)}+ \ldots + m_{\pi (k)})}
\cdot \frac{\alpha (n-m_\ell)+\theta m_{\ell}}{n}\\
&\qquad = \frac{1}{\theta}\sum_{\ell=1}^k \frac{\alpha (n-m_\ell)+\theta m_{\ell}}{n} \sum_{\pi \in S_{k}^{(\ell)}} \prod_{j=2}^{k}\frac{\alpha (m_{\pi(j+1)}+ \ldots + m_{\pi (k)}) +\theta m_{\pi (j)}}{(m_{\pi(j)}+ \ldots + m_{\pi (k)})}.
\end{align*}
We now use the inductive hypothesis for $k-1$ to evaluate the last sum over partitions, and we obtain
\begin{align*}
& \frac{1}{\theta}\sum_{\pi \in S_k} \prod_{j=1}^k \frac{\alpha (m_{\pi(j+1)}+ \ldots + m_{\pi (k)}) +\theta m_{\pi (j)}}{(m_{\pi(j)}+ \ldots + m_{\pi (k)})} = \sum_{\ell=1}^k \frac{\alpha (n-m_\ell)+\theta m_{\ell}}{n} \prod_{i=1}^{k-2} (\theta+i \alpha)\\
& \qquad\qquad =  \prod_{i=1}^{k-2} \cdot \left[k \alpha + \frac{1}{n} (\theta-\alpha) \sum_{j=1}^k m_j \right]
=  \prod_{i=1}^{k-2} (\theta+i \alpha) \cdot [\theta+ \alpha (k-1)] = \prod_{i=1}^{k-1} (\theta+i \alpha)
\end{align*}
thus, \eqref{eq:induction} follows by induction. Now, Equation (9) is an immediate consequence of this equality and the expression of the ordered EPPF (8). \\
\qed

\subsection{Proof of Equation (11)}
\label{app:proof_prop1}

First of all we observe that if $i=n$, then the probability $P_{n}(n; \alpha, \theta)=1$, indeed it is the only species in the sample and it is both the youngest and the oldest. We now assume that $i<n$ and we would like to evaluate the probability of the following event
\[
\Ecr := \{  \text{a species represented $i$ times in a sample of size $n$ has highest weight}\}.
\]
In order to  evaluate $\text{Pr}[\Ecr]$, we denote by $\tilde{\pi}_n$ the random partition generated by the sample of size $n$, which will be characterized by 
a species having frequency $i$, while the remaining $n-i$ observations are partitioned into $K_{n-i} \geq 1$ distinct values with frequencies $w_1, \ldots , w_{k}$. Note that $\tilde{\pi}_n$ does not take into account the ordering of the species. Thus we get
\begin{equation}
\label{eq:ECR}
\text{Pr}[\Ecr]  = \E [\text{Pr}[\Ecr | \tilde{\pi}_n]] =  \E \left[   \frac{\sum_{\pi \in S_{K_{n-i}}}\Phi_{K_{n-i}+1}^{(n)} (i,w_{\pi (1)}, \ldots , w_{\pi (
K_{n-i})})}{\frac{(1-\alpha)_{(i-1)}\prod_{j=1}^{K_{n-i}} (1-\alpha)_{(w_j-1)} }{(\theta+1)_{(n-1)}} \prod_{s=1}^{K_{n-i}+1-1} (\theta+\alpha s)} \right]
\end{equation}
where the denominator is nothing but the EPPF of the PYP, whereas the numerator is the sum of all the order EPPF of the PYP in which the oldest species equals a specific species having frequency $i$. Note that the sum in \eqref{eq:ECR} is over all possible permutations of the remaining $K_{n-i}$ distinct species, thus the probability $\text{Pr}[\Ecr] $ equals
\begin{align*}
& \E \left[\frac{\frac{ (1-\alpha)_{(i-1)}\prod_{j=1}^{K_{n-i}} (1-\alpha)_{(w_j-1)} }{(\theta+1)_{(n-1)}}\sum_{\pi \in S_{K_{n-i}}}  \theta^{-1}  \prod_{j=1}^{K_{n-i}} \frac{\alpha (w_{\pi(j+1)}+ \ldots + w_{\pi (K_{n-i})}) +\theta w_{\pi (j)}}{(w_{\pi(j)}+ \ldots + w_{\pi (K_{n-i})})}
\cdot \frac{\alpha (n-i)+ \theta i }{n}
 }{\frac{(1-\alpha)_{(i-1)}\prod_{j=1}^{K_{n-i}} (1-\alpha)_{(w_j-1)} }{(\theta+1)_{(n-1)}} \prod_{s=1}^{K_{n-i}} (\theta+\alpha s)}  \right] \\
  & =  \E \left[\frac{\sum_{\pi \in S_{K_{n-i}}}  \theta^{-1}  \prod_{j=1}^{K_{n-i}} \frac{\alpha (w_{\pi(j+1)}+ \ldots + w_{\pi (K_{n-i})}) +\theta w_{\pi (j)}}{(w_{\pi(j)}+ \ldots + w_{\pi (K_{n-i})})}
\cdot \frac{\alpha (n-i)+ \theta i }{n}}{ \prod_{s=1}^{K_{n-i}} (\theta+\alpha s)}  \right]  \\
& =  \E \left[ \frac{\prod_{s=1}^{K_{n-i}-1} (\theta+\alpha s)}{\prod_{s=1}^{K_{n-i}} (\theta+\alpha s)}  \cdot \frac{\alpha (n-i)+\theta i}{n} \right]
 = \frac{\alpha n +i (\theta-\alpha)}{n} \cdot \E \left[ \frac{1}{\theta +\alpha K_{n-i}}  \right]
\end{align*}
where we used the explicit expression of the ordered EPPF and the marginalization result \eqref{eq:induction}. Note that the expected value in the last
equation is made w.r.t. the distribution of $K_{n-i}$ (number of distinct types in a sample of size $n-i$). \\ 
\qed

Additionally, in Section~2 of the main manuscript, we reported the graphical representation of $P_n(i; \alpha,\theta)$ for various values of $\alpha < \theta$. Here we report analogous figures for values of $\alpha > \theta$ (left panel of Figure~\ref{fig:Pni_additional}) and for the specific case of $\theta =0$ (right panel of Figure~\ref{fig:Pni_additional}). These are particularly interesting as they display non-increasing patterns, emphasizing the difference and increased flexibility of the ordered PYP compared to the ordered DP of \cite{Don(86)}. 

\begin{figure}[t!]
\centering
\includegraphics[width=0.47\linewidth]{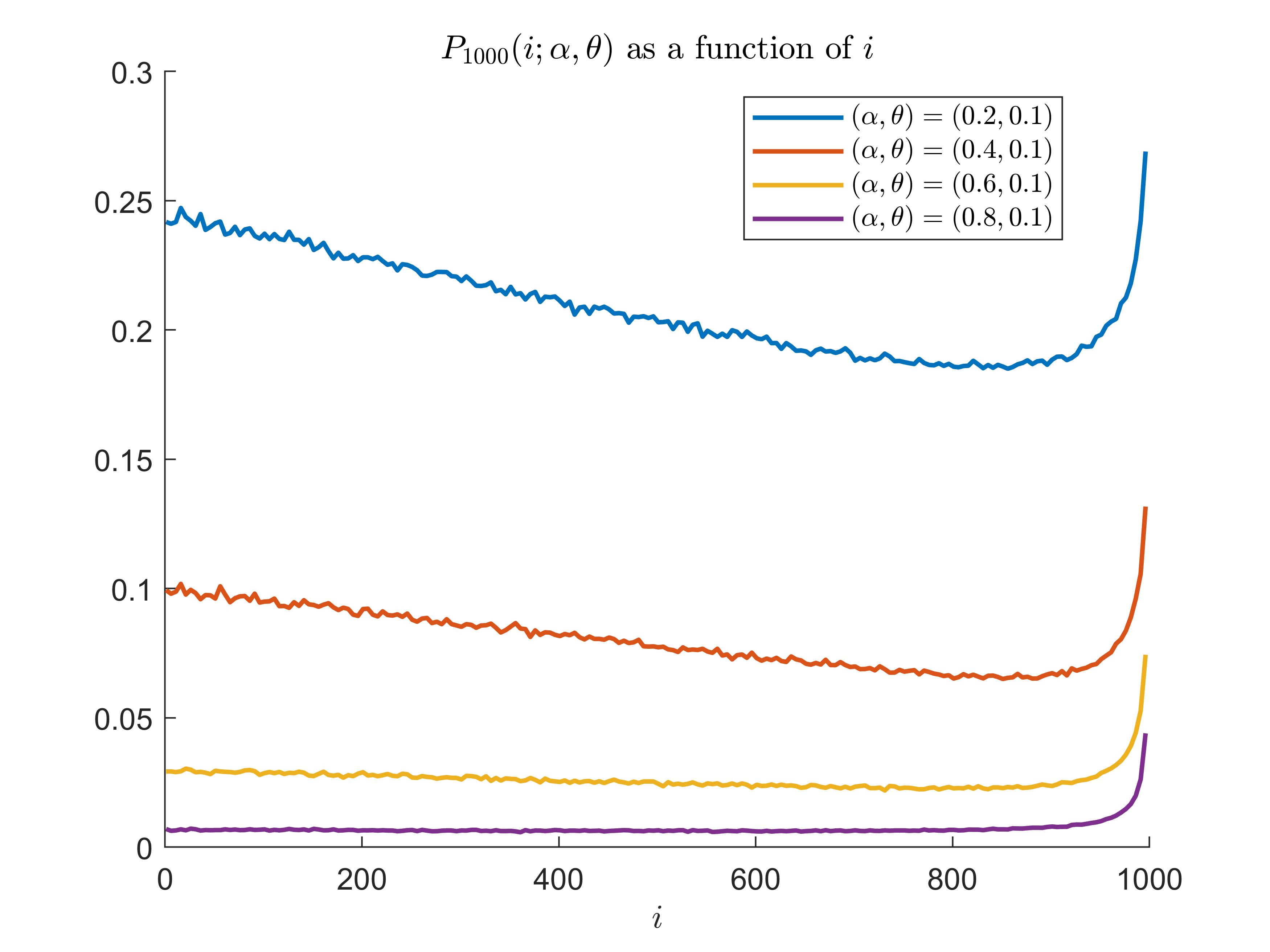}
\includegraphics[width=0.47\linewidth]{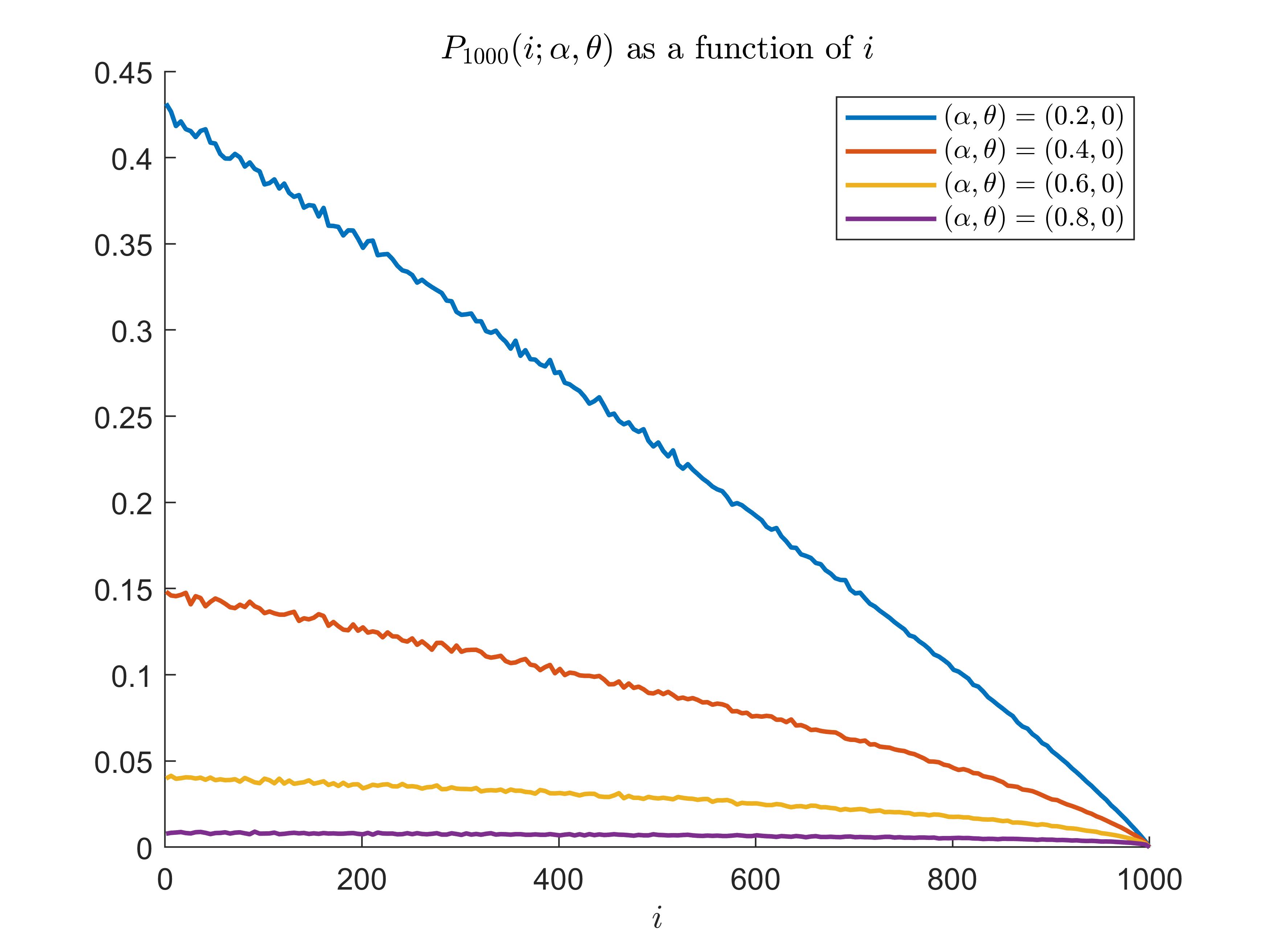} 
\caption{The probability $P_n(i; \alpha, \theta)$, where $n = 1000$, as a function of the frequency $i$, for different values of $\alpha$ and $\theta$. Each panel corresponds to a different $\theta$, from left to right, we have: $\theta = 0.1, 0$.} \label{fig:Pni_additional}
\end{figure}

\subsection{Distributions of the number of distinct species}
\label{app:K}

In this section we remind the distributions of $K_n$ and of $K_{m}^{(n)}$, which has been used in proofs and in the main paper.
The distribution of the number of distinct species $K_n$ is contained in  \citet[Chapter 3]{Pit(06)}. \cite{Pit(06)} has shown that 
that for any $k=1,\ldots,n$
\begin{equation} \label{eq:Kn}
\text{Pr}[K_{n}=k]=\frac{\left(\frac{\theta}{\alpha}\right)_{(k)}}{(\theta)_{(n)}}\mathscr{C}(n,k;\alpha),
\end{equation}
where $\mathscr{C}(n,k;\alpha)=(k!)^{-1}\sum_{0\leq i\leq k}{k\choose i}(-1)^{i}(-i\alpha)_{(n)}$ denotes the generalized factorial coefficient, with the proviso that $\mathscr{C}(0,0;\alpha)=1$ and $\mathscr{C}(m,0;\alpha)=0$ \citep[Chapter 3]{Cha(05)}. See also \citet[Chapter 2]{Pit(06)} for details. 

Now we remind the posterior distribution of $K_{m}^{(n)}=K_{n+m}-K_{n}$, that has been determined in  \citet[Proposition 1]{Lij(07)}. For any $s=0,\ldots,m$ one has 
\begin{equation}
 \label{eq:Kmn}
 \text{Pr}[K_m^{(n)}= s\, |\, K_n=k] = \frac{\left(k+\frac{\theta}{\alpha}\right)_{(s)}}{(\theta+n)_{(m)}} \frac{\Ccr(m,s; \alpha ; -n+k\alpha)}{\alpha^{s}}
\end{equation}
where $\mathscr{C}(n,k;\alpha,\gamma)=(k!)^{-1}\sum_{0\leq i\leq k}{k\choose i}(-1)^{i}(-i\alpha-\gamma)_{(n)}$ denotes the non-centered (shifted) generalized factorial coefficient, with the proviso that $\mathscr{C}(0,0;\alpha,\gamma)=1$ and $\mathscr{C}(n,0;\alpha,\gamma)=(-\gamma)_{(n)}$ \citep[Chapter 3]{Cha(05)}. 
Moreover, the expected value of \eqref{eq:Kmn} can be derived \citep{Fav(09)}  as 
\begin{equation}
 \label{eq:EKmn}
 \mathbb{E}[K_m^{(n)}\, |\, K_n=k] = 
 \left(k+\frac{\theta}{\alpha}\right)\left(\frac{(\theta+n+\alpha)_{(m)}}{(\theta+n)_{(m)}}-1\right).
 \end{equation}

\subsection{Proof of Equation (13)} \label{app:priori}

Recall the ordered EPPF in Equation (8), thus  we can evaluate the following probability
\begin{align*}
& \text{Pr}[M_{1,n} = m_1 , \ldots , M_{r,n} = m_r ,  K_n =k] = \sum_{(\star)}  \binom{n}{m_1, \ldots , m_k}  \Phi_k^{(n)}  (m_1, \ldots , m_k) \\
& \qquad \sum_{(\star)}  \binom{n}{m_1, \ldots , m_{r}, m_{r+1} ,\ldots , m_k}  \Phi_k^{(n)}  (m_1, \ldots , m_{r}, m_{r+1} ,\ldots , m_k)
\end{align*}
where the sum runs over all the vector of positive integers $(m_{r+1}, \ldots , m_{k})$ with $m_i \geq 1$ and 
$m_{r+1}+ \cdots +m_{k}=  n-m_1- \cdots -m_r$. Using the explicit expression of the previous ordered EPPF we get
\begin{align*}
&\text{Pr}[M_{1,n} = m_1 , \ldots , M_{r,n} = m_r ,  K_n =k]  \\
& \qquad = \sum_{(\star)}  \binom{n}{m_1, \ldots , m_{r}, m_{r+1} ,\ldots , m_k}  \frac{\prod_{j=r+1}^{k}  (1-\alpha)_{(m_j-1)}  \prod_{j=1}^{r}  (1-\alpha)_{(m_j-1)}}{(\theta+1)_{(n-1)}}\\
& \qquad\qquad\times  \theta^{-1}  \prod_{j=r+1}^{k}  \frac{\alpha r_{j+1}+\theta m_j}{r_j}
\frac{\prod_{j=1}^r  [\alpha (n-m_1 - \cdots -m_j )+ \theta m_j]}{n (n-m_1) \cdots (n-m_1 -\cdots -m_{r-1})}\\
& \quad =  \frac{n!}{(\theta+1)_{(n-1)}}\frac{1}{(n-m_1 - \cdots -m_r)!}  
\frac{\prod_{j=1}^r  [\alpha (n-m_1 - \cdots -m_j )+ \theta m_j]}{n (n-m_1) \cdots (n-m_1 -\cdots -m_{r-1})}  \cdot \prod_{j=1}^r  \frac{(1-\alpha)_{(m_j-1)}}{m_j !}  \\
& \qquad\qquad \times  \sum_{(\star)}  \binom{n-m_1 - \cdots -m_r}{m_{r+1}, \ldots , m_{k}}  \theta^{-1} \prod_{j=r+1}^{k}  (1-\alpha)_{(m_j-1)}
\prod_{j=r+1}^{k}  \frac{\alpha (m_{j+1} + \cdots+m_{k}) +\theta m_j }{m_j +\cdots +m_k}.
\end{align*}
If we divide the last sum by the coefficient $(\theta+1)_{(n-m_1 - \ldots - m_r-1)}$, we can recognize the probability 
$\text{Pr}[K_{n-m_1 - \ldots - m_r}= k-r]$, thus  one has
\begin{align*}
&\text{Pr}[M_{1,n} = m_1 , \ldots , M_{r,n} = m_r ,  K_n =k ] = \frac{n!}{(\theta+1)_{(n-1)}}\frac{(\theta+1)_{(n-m_1 - \ldots - m_r-1)}}{(n-m_1 - \cdots -m_r)!}  \\& \qquad \times
\frac{\prod_{j=1}^r  [\alpha (n-m_1 - \cdots -m_j )+ \theta m_j]}{n (n-m_1) \cdots (n-m_1 -\cdots -m_{r-1})}  \cdot \prod_{j=1}^r  \frac{(1-\alpha)_{(m_j-1)}}{m_j !} \text{Pr}[K_{n-m_1 - \cdots - m_r}= k-r]  .
\end{align*}
Now we can exploit \eqref{eq:Kn} to get
\begin{align*}
&\text{Pr}[M_{1,n} = m_1 , \ldots , M_{r,n} = m_r ,  K_n =k] = \frac{n!}{(\theta+1)_{(n-1)}}\frac{ \Ccr (n-m_1 - \ldots - m_r,k-r;\alpha)}{(n-m_1 - \cdots -m_r)!}  \\& \qquad \times
\frac{\prod_{j=1}^r  [\alpha (n-m_1 - \cdots -m_j )+ \theta m_j]}{n (n-m_1) \cdots (n-m_1 -\cdots -m_{r-1})}  \cdot \prod_{j=1}^r  \frac{(1-\alpha)_{(m_j-1)}}{m_j !} \frac{\prod_{i=1}^{k-r-1} (\theta+i \alpha)}{\alpha^{k-r} } .
\end{align*}
By dividing the previous expression by the probability $\text{Pr}[K_n=k]$ in \eqref{eq:Kn}, one may now obtain the following conditional distribution 
\begin{align} \label{eq:prior_frequenze_ordinate_cond}
& \text{Pr} [M_{1,n} = m_1 , \ldots , M_{r,n} = m_r\, |\,  K_n =k]\\
&\notag\quad= {n\choose m_{1},\ldots,m_{r},n-|\mathbf{m}|_{1:r}}\frac{\displaystyle\prod_{j=1}^r \frac{[\alpha (n-|\mathbf{m}|_{1:j} )+ \theta m_j](1-\alpha)_{(m_j-1)}}{n-|\mathbf{m}|_{1:j-1}}}{\displaystyle\frac{\Ccr (n,k; \alpha)}{\left(\frac{\theta}{\alpha}+k-r\right)_{(r)}\Ccr (n-|\mathbf{m}|_{1:r},k-r;\alpha)}}.
\end{align}
In order to obtain Equation (13), note that:
\begin{align*}
\text{Pr}[M_{1,n} = m_1 , \ldots , M_{r,n} = m_r ,  K_n \geq r] = \sum_{k=r}^s \text{Pr}[M_{1,n} = m_1 , \ldots , M_{r,n} = m_r ,  K_n =k] 
\end{align*}
 where $s= r+n-m_1- \cdots-m_r$, thus, we get
 \begin{align*}
&\text{Pr}[M_{1,n} = m_1 , \ldots , M_{r,n} = m_r ,  K_n \geq r] = \sum_{k=r}^s 
\frac{n!}{(\theta+1)_{(n-1)}}\frac{(\theta+1)_{(n-m_1 - \ldots - m_r-1)}}{(n-m_1 - \cdots -m_r)!}  \\& \qquad \times
\frac{\prod_{j=1}^r  [\alpha (n-m_1 - \cdots -m_j )+ \theta m_j]}{n (n-m_1) \cdots (n-m_1 -\cdots -m_{r-1})}  \cdot \prod_{j=1}^r  \frac{(1-\alpha)_{(m_j-1)}}{m_j !} \text{Pr}[K_{n-m_1 - \cdots - m_r}= k-r]
\end{align*}
where we observe that the sum over $k$ of the probability $\text{Pr}[K_{n-m_1 - \cdots - m_r}= k-r]$ is equal to $1$, as a consequence we obtain
 \begin{align*}
&\text{Pr}[M_{1,n} = m_1 , \ldots , M_{r,n} = m_r ,  K_n \geq r] =
\frac{n!}{(\theta+1)_{(n-1)}}\frac{(\theta+1)_{(n-m_1 - \ldots - m_r-1)}}{(n-m_1 - \cdots -m_r)!}  \\& \qquad \times
\frac{\prod_{j=1}^r  [\alpha (n-m_1 - \cdots -m_j )+ \theta m_j]}{n (n-m_1) \cdots (n-m_1 -\cdots -m_{r-1})}  \cdot \prod_{j=1}^r  \frac{(1-\alpha)_{(m_j-1)}}{m_j !}
\end{align*}
and the thesis follows.\\
\qed

\subsection{Proof of Theorem 1} \label{app:posterior}

We first focus on the proof of Equation (17). To this end, evaluate
\begin{align*}
&\text{Pr}[A_r, W_{1,n+m} = w_1 , \ldots , W_{r,n+m} = w_r ,  K_m^{(n)} = s | K_n=k , \mathbf{M}_{n}=\mathbf{m}]\\
&\qquad = \frac{\text{Pr}[A_r, W_{1,n+m} = w_1 , \ldots , W_{r,n+m} = w_r ,  K_m^{(n)} = s , K_n=k , M_{i,n}= m_i, \; i =1, \ldots , k]}{
\text{Pr}[ K_n=k , M_{i,n}= m_i, \; i =1, \ldots , k]}
\end{align*}
where the numerator may be evaluated along similar lines as in Section \ref{app:priori} of the Supplementary Material. In fact one has
\begin{align*}
& \text{Pr}[A_r, W_{1,n+m} = w_1 , \ldots , W_{r,n+m} = w_r ,  K_m^{(n)} = s , K_n=k , M_{i,n}= m_i, \; i =1, \ldots , k]\\
& \qquad = \sum_{(\star)}  \binom{n}{n_{r+1}, \ldots , n_{k+s}}  
\binom{m}{w_{1}, \ldots , w_{r}, w_{r+1} , \ldots , w_{k+s}} \frac{\prod_{j=r+1}^{k+s} (1-\alpha)_{(w_{j}+n_{j}-1)}}{
(\theta+1)_{(n+m-1)}} \\
& \qquad\qquad \times  \prod_{j=1}^r (1-\alpha)_{(w_j-1)} \theta^{-1} \prod_{j=r+1}^{k+s} \frac{\alpha (w_{j+1}+n_{j+1}+ \cdots + w_{k+s}+n_{k+s})+\theta (w_j +n_j)}{w_{j}+n_{j}+ \cdots + w_{k+s}+n_{k+s}}\\
& \qquad\qquad\qquad \times \frac{\prod_{j=1}^r (\alpha (n+m-w_1 - \cdots -w_j) +\theta w_j)}{(n+m) (n+m-w_1) \cdots (n+m-w_1- \cdots - w_{r-1})}
\end{align*}
where the sum runs over all the vectors $(w_{r+1}, \ldots , w_{k+s})$ and $(n_{r+1}, \ldots , n_{k+s})$ such that $w_i \geq 0$ and $w_{r+1}+\cdots+w_{k+s}=m-w_1- \cdots-w_r$, whereas $n_{r+1}+\cdots+n_{k+s}= n$ and $n_{i_j}= m_j$ for $j=1, \ldots, k$ and $i_1 < \cdots < i_k$ the other values of $n_i$ are equal to $0$; 
finally $w_i+n_i \geq 1$ for any $i=r+1, \ldots , k+s$. Thus
\begin{align*}
&\text{Pr}[A_r, W_{1,n+m} = w_1 , \ldots , W_{r,n+m} = w_r ,  K_m^{(n)} = s , K_n=k , M_{i,n}= m_i, \; i =1, \ldots , k]\\
& \qquad = \frac{\prod_{j=1}^r (\alpha (n+m-w_1 - \cdots -w_j) +\theta w_j)}{(n+m) (n+m-w_1) \cdots (n+m-w_1- \cdots - w_{r-1})} \prod_{j=1}^r \frac{(1-\alpha)_{(w_j-1)}}{w_j !}\\
& \qquad \times \frac{m!}{ (m-w_1 - \cdots -w_r)! (\theta+1)_{(n+m-1)}}  
\sum_{(\star)} \binom{n}{n_{r+1}, \ldots , n_{k+s}} \binom{m-w_1 - \cdots -w_r}{w_{r+1}, \ldots , w_{k+s}}   \\
& \qquad\qquad\times  \theta^{-1}\prod_{j=r+1}^{k+s} (1-\alpha)_{(w_j+n_j-1)} \prod_{j=r+1}^{k+s} \frac{\alpha (w_{j+1}+n_{j+1}+ \cdots + w_{k+s}+n_{k+s})+\theta (w_j +n_j)}{w_{j}+n_{j}+ \cdots + w_{k+s}+n_{k+s}}
\end{align*}
if we divide the last sum by the coefficient $(\theta+1)_{(n+m-w_1- \cdots -w_r -1)}$ we obtain the following 
\begin{equation}
\label{eq:su_post}
\begin{split}
& \frac{1}{(\theta+1)_{(n+m-w_1- \cdots -w_r -1)}}\sum_{(\star)}  \binom{n}{n_{r+1}, \ldots , n_{k+s}} \binom{m-w_1 - \cdots -w_r}{w_{r+1}, \ldots , w_{k+s}}  \\
& \quad \times  \theta^{-1}\prod_{j=r+1}^{k+s} (1-\alpha)_{(w_j+n_j-1)} \\
& \qquad\qquad \times\prod_{j=r+1}^{k+s} \frac{\alpha (w_{j+1}+n_{j+1}+ \cdots + w_{k+s}+n_{k+s})+\theta (w_j +n_j)}{w_j+n_j+ \cdots + w_{k+s}+n_{k+s}}\\
 & \qquad  = \text{Pr}[K_{m-w_1- \cdots -w_r}^{(n)}= s-r |  K_n=k , \mathbf{M}_{n}=\mathbf{m}] \\
 & \qquad\qquad \times\text{Pr}[K_n=k ,\mathbf{M}_{n}=\mathbf{m}].
\end{split}
\end{equation}
Substituting \eqref{eq:su_post} in the expression of the joint probability under study, we obtain 
\begin{align*}
&\text{Pr}[A_r, W_{1,n+m} = w_1 , \ldots , W_{r,n+m} = w_r ,  K_m^{(n)} = s , K_n=k , M_{i,n}= m_i, \; i =1, \ldots , k]\\
& \qquad = \frac{\prod_{j=1}^r (\alpha (n+m-w_1 - \cdots -w_j) +\theta w_j)}{(n+m) (n+m-w_1) \cdots (n+m-w_1- \cdots - w_{r-1})} \prod_{j=1}^r \frac{(1-\alpha)_{(w_j-1)}}{w_j !}\\
& \qquad\qquad \times \frac{m! (\theta+1)_{(n+m-w_1- \cdots -w_r -1)}}{ (m-w_1 - \cdots -w_r)! (\theta+1)_{(n+m-1)}}  \\
& \qquad\qquad\qquad \times
 \text{Pr}[K_{m-w_1- \cdots -w_r}^{(n)}= s-r |  K_n=k , \mathbf{M}_{n}=\mathbf{m}] \\
 & \qquad\qquad\qquad\qquad \times\text{Pr} [K_n=k , \mathbf{M}_{n}=\mathbf{m}].
\end{align*}
As a consequence:
\begin{align*}
&\text{Pr}[A_r, W_{1,n+m} = w_1 , \ldots , W_{r,n+m} = w_r ,  K_m^{(n)} = s | K_n=k , \mathbf{M}_{n}=\mathbf{m}]\\
& \qquad = \frac{\prod_{j=1}^r (\alpha (n+m-w_1 - \cdots -w_j) +\theta w_j)}{(n+m) (n+m-w_1) \cdots (n+m-w_1- \cdots - w_{r-1})} \prod_{j=1}^r \frac{(1-\alpha)_{(w_j-1)}}{w_j !}\\
& \qquad\qquad \times \frac{m! (\theta+1)_{(n+m-w_1- \cdots -w_r -1)}}{ (m-w_1 - \cdots -w_r)! (\theta+1)_{(n+m-1)}}  \\
& \qquad\qquad\qquad \times
 \text{Pr}[K_{m-w_1- \cdots -w_r}^{(n)}= s-r |  K_n=k , \mathbf{M}_{n}=\mathbf{m}]
\end{align*}
where the last probability can be evaluated resorting to the posterior distribution of $K_{m-w_1- \cdots -w_r}^{(n)}$ in \eqref{eq:Kmn}. Summing over $s= r, \ldots , r+m-w_1-\cdots-w_r$ the previous expression, we obtain Equation (17).\\

We now move to the proof of Equation (18). Proceeding along the same lines as in the first part of the proof, we have
\begin{align*}
&\text{Pr}[B_r, W_{1,n+m} = w_1+m_1 , \ldots , W_{r,n+m} = w_r+m_r ,  K_m^{(n)} = s | K_n=k , \mathbf{M}_{n}=\mathbf{m}]\\
&\quad = \frac{\text{Pr}[B_r, W_{1,n+m} = w_1+m_1 , \ldots , W_{r,n+m} = w_r +m_r,  K_m^{(n)} = s , K_n=k , \mathbf{M}_{n}=\mathbf{m}]}{
\text{Pr}[ K_n=k , \mathbf{M}_{n}=\mathbf{m}]}
\end{align*}
where the numerator may be evaluated along similar lines as in the proof of Theorem 1. In fact, as for the numerator,  one has
\begin{align*}
& \text{Pr} [B_r, W_{1,n+m} = w_1+m_1 , \ldots , W_{r,n+m} = w_r +m_r,  K_m^{(n)} = s , K_n=k , \mathbf{M}_{n}=\mathbf{m}]\\
& \qquad = \sum_{(\star)}  \binom{n}{m_1, \ldots , m_r, n_{r+1}, \ldots , n_{k+s}}  
\binom{m}{w_1, \ldots , w_r,  w_{r+1}, \ldots , w_{k+s}}
\frac{\prod_{j=r+1}^{k+s} (1-\alpha)_{(w_j+n_j-1)}}{
(\theta+1)_{(n+m-1)}} \\
& \qquad\qquad \times  \prod_{j=1}^r (1-\alpha)_{(w_j+m_j-1)} \theta^{-1} \prod_{j=r+1}^{k+s} \frac{\alpha (w_{j+1}+n_{j+1}+ \cdots + w_{k+s}+n_{k+s})+\theta (w_j +n_j)}{w_j+n_j+ \cdots + w_{k+s}+n_{k+s}}\\
& \qquad\qquad\qquad \times \frac{\prod_{j=1}^r (\alpha (n+m-w_1 -m_1- \cdots -w_j-m_j) +\theta (w_j+m_j))}{(n+m) (n+m-w_1-m_1) \cdots (n+m-w_1-m_1- \cdots - w_{r-1}-m_{r-1})}
\end{align*}
where the sum runs over all the vectors $(w_{r+1}, \ldots , w_{k+s})$ and $(n_{r+1}, \ldots , n_{k+s})$ such that $w_i \geq 0$ and $w_{r+1}+\cdots+w_{k+s}=m-w_1- \cdots-w_r$, whereas $n_{r+1}+\cdots+n_{k+s}= n-m_1- \cdots-m_r$ and $n_{i_j}= m_j$ for $j=1, \ldots, k-r$ and $i_1 < \cdots < i_{k-r}$ the other values of $n_i$ are equal to $0$; 
finally $w_i+n_i \geq 1$ for any $i=r+1, \ldots , k+s$. Thus
\begin{align*}
& \text{Pr}[B_r, W_{1,n+m} = w_1+m_1 , \ldots , W_{r,n+m} = w_r+m_r ,  K_m^{(n)} = s , K_n=k , \mathbf{M}_{n}=\mathbf{m}]\\
& \qquad = \frac{\prod_{j=1}^r (\alpha (n+m-w_1-m_1 - \cdots -w_j-m_j) +\theta (w_j+m_j))}{(n+m) (n+m-w_1-m_1) \cdots (n+m-w_1-m_1- \cdots - w_{r-1}-m_{r-1})} \\
& \qquad \times  \prod_{j=1}^r \frac{(1-\alpha)_{(w_j+m_j-1)}}{w_j ! m_j!} \frac{m!n!}{ (m-w_1 - \cdots -w_r)! (n-m_1-\cdots-m_r)! (\theta+1)_{(n+m-1)}} 
  \\
& \qquad\qquad \times \sum_{(\star)} \binom{n-m_1-\cdots-m_r}{n_{r+1}, \ldots , n_{k+s}} \binom{m-w_1 - \cdots -w_r}{w_{r+1}, \ldots , w_{k+s}} \theta^{-1}\prod_{j=r+1}^{k+s} (1-\alpha)_{(w_j+n_j-1)}  \\
& \qquad\qquad\qquad\times  \prod_{j=r+1}^{k+s} \frac{\alpha (w_{j+1}+n_{j+1}+ \cdots + w_{k+s}+n_{k+s})+\theta (w_j +n_j)}{w_j+n_j+ \cdots + w_{k+s}+n_{k+s}}
\end{align*}
if we divide the last sum by the coefficient $(\theta+1)_{(n+m-w_1-m_1- \cdots -w_r-m_r -1)}$ we obtain the following 
\begin{equation}
\label{eq:su_post_old}
\begin{split}
& \frac{1}{(\theta+1)_{(n+m-w_1-m_1- \cdots -w_r-m_r -1)}}\sum_{(\star)}  \binom{n-m_1-\cdots-m_r}{n_{r+1}, \ldots , n_{k+s}} \binom{m-w_1 - \cdots -w_r}{w_{r+1}, \ldots , w_{k+s}}  \\
& \quad \times  \theta^{-1}\prod_{j=r+1}^{k+s} (1-\alpha)_{(w_j+n_j-1)}\\
& \qquad\qquad \times\prod_{j=r+1}^{k+s} \frac{\alpha (w_{j+1}+n_{j+1}+ \cdots + w_{k+s}+n_{k+s})+\theta (w_j +n_j)}{w_j+n_j+ \cdots + w_{k+s}+n_{k+s}}\\
 & \qquad  =  \text{Pr}[K_{m-w_1- \cdots -w_r}^{(n-m_1-\cdots-m_r)}= s |  K_{n-m_1-\cdots-m_r}=k -r, M_{i,n}= m_i, \; i =r+1, \ldots , k] \\
 & \qquad\qquad \times\text{Pr}[ K_{n-m_1-\cdots-m_r}=k-r , M_{i,n}= m_i, \; i =r+1, \ldots , k].
\end{split}
\end{equation}
Substituting \eqref{eq:su_post_old} in the expression of the joint probability under study, we obtain 
\begin{align*}
& \text{Pr}[B_r, W_{1,n+m} = w_1+m_1 , \ldots , W_{r,n+m} = w_r+m_r ,  K_m^{(n)} = s , K_n=k , \mathbf{M}_{n}=\mathbf{m}]\\
& \qquad = \frac{\prod_{j=1}^r (\alpha (n+m-w_1-m_1 - \cdots -w_j-m_j) +\theta (w_j+m_j))}{(n+m) (n+m-w_1-m_1) \cdots (n+m-w_1-m_j- \cdots - w_{r-1}-m_{r-1})} \\
& \qquad\quad \times \prod_{j=1}^r \frac{(1-\alpha)_{(w_j+m_j-1)}}{w_j !m_j!} \cdot
 \frac{m!n!(\theta+1)_{(n+m-w_1-m_1-\cdots-w_r-m_r-1)} }{ (m-w_1 - \cdots -w_r)! (n-m_1-\cdots-m_r)!(\theta+1)_{(n+m-1)}}  \\
& \qquad\quad\quad \qquad\times
 \text{Pr} [K_{m-w_1- \cdots -w_r}^{(n-m_1-\cdots-m_r)}= s |  K_{n-m_1-\cdots-m_r}=k-r , M_{i,n}= m_i, \; i =r+1, \ldots , k] \\
 & \qquad\quad\quad\quad \qquad\times\text{Pr}[ K_{n-m_1-\cdots-m_r}=k-r , M_{i,n}= m_i, \; i =r+1, \ldots , k].
\end{align*}
As a consequence:
\begin{align*}
&\text{Pr}[B_r, W_{1,n+m} = w_1+m_1 , \ldots , W_{r,n+m} = w_r+m_r ,  K_m^{(n)} = s | K_n=k , \mathbf{M}_{n}=\mathbf{m}]\\
& \qquad = \frac{\prod_{j=1}^r (\alpha (n+m-w_1-m_1 - \cdots -w_j-m_j) +\theta (w_j+m_j))}{(n+m) (n+m-w_1-m_1) \cdots (n+m-w_1-m_j- \cdots - w_{r-1}-m_{r-1})} \\
& \qquad\quad \times \prod_{j=1}^r \frac{(1-\alpha)_{(w_j+m_j-1)}}{w_j !m_j!} \cdot
 \frac{m!n!(\theta+1)_{(n+m-w_1-m_1-\cdots-w_r-m_r-1)} }{ (m-w_1 - \cdots -w_r)! (n-m_1-\cdots-m_r)!(\theta+1)_{(n+m-1)}}  \\
& \qquad\quad\quad \qquad\times
 \text{Pr}[K_{m-w_1- \cdots -w_r}^{(n-m_1-\cdots-m_r)}= s |  K_{n-m_1-\cdots-m_r}=k-r , M_{i,n}= m_i, \; i =r+1, \ldots , k] \\
 & \qquad\quad\quad\quad \qquad\times\frac{\text{Pr}[K_{n-m_1-\cdots-m_r}=k-r , M_{i,n}= m_i, \; i =r+1, \ldots , k]}{
\text{Pr} [ K_{n}=k , M_{i,n}= m_i, \; i =1, \ldots , k] }.
\end{align*}
Thus we have to evaluate the last ratio in the previous expression, this calculation can be easily addressed resorting to the available expression of the ordered EPPF for the PYP (8). Indeed we get:
\begin{align*}
&\frac{\text{Pr} [K_{n-m_1-\cdots-m_r}=k-r , M_{i,n}= m_i, \; i =r+1, \ldots , k]}{
\text{Pr} [ K_{n}=k , M_{i,n}= m_i, \; i =1, \ldots , k] }= \frac{\binom{n-m_1-\cdots-m_r}{m_{r+1}, \ldots , m_k}}{\binom{n}{m_1, \ldots , m_k}}  \\
& \qquad  \times\frac{(\theta+1)_{(n-1)}}{(\theta+1)_{(n-m_1-\cdots-m_r-1)}} \cdot \frac{1}{\prod_{j=1}^r (1-\alpha)_{(m_j-1)}}
\cdot \frac{n(n-m_1)\cdots (n-m_1-\cdots-m_{r-1})}{\prod_{j=1}^r (\alpha (n-m_1-\cdots-m_j )+\theta m_j)}.
\end{align*}
Substituting the previous equality in the expression of the conditional probability under study we get:
\begin{align*}
&\text{Pr} [B_r, W_{1,n+m} = w_1+m_1 , \ldots , W_{r,n+m} = w_r+m_r ,  K_m^{(n)} = s | K_n=k , \mathbf{M}_{n}=\mathbf{m}]\\
& \qquad = \frac{\prod_{j=1}^r (\alpha (n+m-w_1-m_1 - \cdots -w_j-m_j) +\theta (w_j+m_j))}{(n+m) (n+m-w_1-m_1) \cdots (n+m-w_1-m_j- \cdots - w_{r-1}-m_{r-1})} \\
& \qquad\quad \times \prod_{j=1}^r \frac{(1-\alpha)_{(w_j+m_j-1)}}{w_j !(1-\alpha)_{(m_j-1)}} \cdot
 \frac{m!(\theta+1)_{(n+m-w_1-m_1-\cdots-w_r-m_r-1)} (\theta+1)_{(n-1)}}{ (m-w_1 - \cdots -w_r)!(\theta+1)_{(n+m-1)} (\theta+1)_{(n-m_1-\cdots-m_r-1)}}  \\
& \qquad\quad\quad \qquad\times \text{Pr} [K_{m-w_1- \cdots -w_r}^{(n-m_1-\cdots-m_r)}= s |  K_{n-m_1-\cdots-m_r}=k-r , M_{i,n}= m_i, \; i =r+1, \ldots , k]  \\
 & \qquad\quad\quad\quad \qquad\times \frac{n(n-m_1)\cdots (n-m_1-\cdots-m_{r-1})}{\prod_{j=1}^r (\alpha (n-m_1-\cdots-m_j )+\theta m_j)}  .
\end{align*}
Summing overall the possible values of $s \in \{ 0, \ldots , m-w_1-\cdots -w_r\}$ we finally get the expression Equation (18).\\
\qed

\subsection{Proof of Equation (22)} \label{app:expectation}

We focus on the proof of Equation (22) by evaluating the two expected values
\begin{equation} \label{eq:EW1A}
    \E [W_{1, n+m} \indic_{A_1}| K_n =k , \mathbf{M}_{n}=\mathbf{m}] 
\end{equation}
and
\begin{equation} \label{eq:EW1B}
    \E [W_{1, n+m} \indic_{B_1}| K_n =k , \mathbf{M}_{n}=\mathbf{m}].
\end{equation}
To calculate \eqref{eq:EW1A}, we need to evaluate the posterior expected value of $W_{1,n+m}$ on the event $A_1$:
\begin{align}
&\E [W_{1, n+m} \indic_{A_1}| K_n =k , \mathbf{M}_{n}=\mathbf{m}] \nonumber\\
 &\qquad = \sum_{w=1}^m
w \binom{m}{w}\frac{\alpha (n+m-w) +\theta w}{n+m}  \cdot  \frac{(1-\alpha)_{(w-1)} (\theta+1)_{(n+m-w-1)}}{(\theta+1)_{(n+m-1)}} \nonumber\\
& \qquad =  \frac{m}{(n+m) (\theta+1)_{(n+m-1)}} \nonumber \\
& \qquad\qquad \times \sum_{w=0}^{m-1} \binom{m-1}{w} [\alpha  (n+m-w-1)+\theta (w+1)]
(1-\alpha)_{(w)} (\theta+1)_{(n+m-w-2)}   \nonumber\\
& \qquad =  \frac{m}{(n+m)\Gamma (\theta+n+m)\Gamma (1-\alpha)}  (S_1+S_2)  \label{eq:WA_S12}
\end{align}
where $S_1$ and $S_2$ are the two sums defined as follows
\begin{align}
\label{eq:sum1}
S_1 := \sum_{w=0}^{m-1} \binom{m-1}{w} [\alpha (n+m-1) +\theta] \Gamma (w-\alpha+1) \Gamma (\theta+n+m-q-1)\\
S_2 :=  \sum_{w=0}^{m-1}  \binom{m-1}{w} w (\theta-\alpha) \Gamma (w-\alpha+1) \Gamma (\theta+n+m-w-1).
\label{eq:sum2}
\end{align} 
We now concentrate on the two sums separately. In order to evaluate the first one \eqref{eq:sum1} we exploit the integral representation of the beta function:
\begin{align*}
S_1 & =  [\alpha (n+m-1)+\theta] \Gamma (\theta+n+m-\alpha) \sum_{w=0}^{m-1} \binom{m-1}{w} B (w-\alpha+1, \theta+n+m-w-1)  \\
& = [\alpha (n+m-1)+\theta] \Gamma (\theta+n+m-\alpha) \\
& \qquad\qquad \times \int_0^1  x^{-\alpha}  (1-x)^{\theta+n+m-2}
 \sum_{w=0}^{m-1}  \binom{m-1}{w}   \frac{x^w}{(1-x)^w} \de x \\
 & =  [\alpha (n+m-1)+\theta] \Gamma (\theta+n+m-\alpha)   \int_0^1  x^{-\alpha} (1-x)^{\theta+n+m-2} 
 \Big( 1+\frac{x}{1-x} \Big)^{m-1}  \de x  \\
 & =  [\alpha (n+m-1)+\theta] \Gamma (\theta+n+m-\alpha)  \Gamma (\theta+n+m-\alpha)   B (1-\alpha, \theta+n) .
\end{align*}
The evaluation of the sum \eqref{eq:sum2} proceeds in a similar fashion. We first observe that for $w=0$ the summand 
in $S_2$ is zero, thus we can write:
\begin{align*}
S_2 & = (\theta-\alpha) (m-1) \sum_{w=0}^{m-2} \binom{m-2}{w} \Gamma (w-\alpha+2) \Gamma (\theta+n+m-w-2)\\
& = (\theta-\alpha) (m-1)  \Gamma (\theta+n+m-\alpha) \sum_{w=0}^{m-2} \binom{m-2}{w} B(w-\alpha+2, \theta+n+m-w-2)\\
& = (\theta-\alpha) (m-1)  \Gamma (\theta+n+m-\alpha) \int_0^1 \sum_{w=0}^{m-2} \binom{m-2}{w} x^{w-\alpha+1} (1-x)^{\theta+n+m-w-3}\de x\\
&  = (\theta-\alpha) (m-1)  \Gamma (\theta+n+m-\alpha) \int_0^1  x^{-\alpha+1} (1-x)^{\theta+n+m-3} \Big(  1+\frac{x}{1-x}\Big)^{m-2}  \de x\\
& =(\theta-\alpha) (m-1)  \Gamma (\theta+n+m-\alpha) B (2-\alpha , \theta+n).
\end{align*}
We can now put together the expressions we determined for $S_1$ and $S_2$ in \eqref{eq:WA_S12} to get
\begin{align*}
&\E [W_{1, n+m} \indic_{A_1}| K_n =k , \mathbf{M}_{n}=\mathbf{m}]  \\
& \qquad  =  \frac{m   \Gamma (\theta+n+m-\alpha) \Gamma (\theta+n)}{(n+m)  \Gamma (\theta+n+1-\alpha) \Gamma (\theta+n+m)} \\
& \qquad\qquad \qquad\times \left\{  \alpha (n+m-1)+\theta +\frac{(\theta-\alpha) (m-1)(1-\alpha)}{\theta+n+1-\alpha}  \right\}\\
& \qquad =  \frac{m}{n+m} \frac{(\theta+n-\alpha)_{(m)}}{(\theta+n)_{(m)} (\theta+n-\alpha)}  \left\{ \alpha n+\theta
+\alpha (m-1) +\frac{(\theta-\alpha)(m-1)(1-\alpha)}{\theta+n+1-\alpha} \right\} \\
&  \qquad = \frac{m}{n+m} \frac{(\theta+n-\alpha)_{(m)}}{(\theta+n)_{(m)} (\theta+n-\alpha)} \cdot (\theta+n \alpha) 
\frac{\theta+n+m-\alpha}{\theta+n+1-\alpha}
\end{align*}
where the last equality follows from straightforward calculations. Thus the expression of \eqref{eq:EW1A} follows
\begin{equation} \label{eq:EWA1_final}
\E [W_{1, n+m} \indic_{A_1}| K_n =k , \mathbf{M}_{n}=\mathbf{m}] = \frac{m}{n+m}\cdot \frac{\theta+n\alpha}{\theta+n+1-\alpha} \frac{(\theta+n+1-\alpha)_{(m)}}{(\theta+n)_{(m)}}.
\end{equation}

We now concentrate on the evaluation  of \eqref{eq:EW1B}, which follows in a similar manner.
\begin{align*}
&\E [W_{1, n+m} \indic_{B_1}| K_n =k , \mathbf{M}_{n}=\mathbf{m}] \\
 &\qquad = \sum_{w=m_1}^{m_1+m} w  \frac{\alpha (n+m-w) +\theta w}{n+m} \cdot  \frac{(1-\alpha)_{(w-1)} (\theta+1)_{(n+m-w-1)}}{(\theta+1)_{(n+m-1)}  } \\
 & \qquad \qquad \qquad \times \frac{n}{\alpha (n-m_1) +\theta m_1} \binom{m}{w-m_1} \frac{(\theta+1)_{(n-1)}}{(1-\alpha)_{(m_1-1)} (\theta+1)_{(n-m_1-1)}}\\
 &  \qquad =  \sum_{w=0}^{m} (w+m_1)  \frac{\alpha (n+m-w-m_1) +\theta (w+m_1)}{n+m} \cdot  \frac{(1-\alpha)_{(w+m_1-1)} (\theta+1)_{(n+m-w-m_1-1)}}{(\theta+1)_{(n+m-1)}  } \\
 & \qquad \qquad \qquad \times \frac{n}{\alpha (n-m_1) +\theta m_1} \binom{m}{w} \frac{(\theta+1)_{(n-1)}}{(1-\alpha)_{(m_1-1)} (\theta+1)_{(n-m_1-1)}}\\
 & \qquad =  \frac{n  }{(n+m) (\alpha (n-m_1) +\theta m_1)}  \cdot  \frac{(\theta+1)_{(n-1)}}{(\theta+1)_{(n+m-1)}
 (1-\alpha)_{(m_1-1)}  (\theta+1)_{(n-m_1-1)}}  \\
 & \qquad\qquad\qquad \times  \frac{1}{\Gamma (1-\alpha) \Gamma (\theta+1)}  \sum_{w=0}^m  [\alpha (n+m-w-m_1)
 +\theta  (w+m_1)] (w+m_1)  \binom{m}{w}\\
& \qquad\qquad\qquad\qquad \times  \Gamma (w+m_1-\alpha) \Gamma (\theta+n+m-w-m_1).
\end{align*}
With some simple rearrangements of the terms the sum in the previous expression can be written as
\begin{equation}
\label{eq:sommone}
\begin{split}
&\sum_{w=0}^m  [\alpha (n+m-w-m_1)
+\theta  (w+m_1)] (w+m_1)\\
& \qquad\qquad \times  \binom{m}{w} \Gamma (w+m_1-\alpha) \Gamma (\theta+n+m-w-m_1)\\
& \qquad\qquad\qquad = [\alpha (n+m-m_1)+\theta m_1] m_1  R_1   \\
& \qquad\qquad\qquad\qquad + [m_1 (\theta-\alpha)+\alpha (n+m-m_1) +\theta m_1] R_2  +
(\theta-\alpha)  R_3
\end{split}
\end{equation}
where we have defined
\begin{align}
R_1 & :=  \sum_{w=0}^m  \binom{m}{w} \Gamma (w+m_1-\alpha) \Gamma (\theta+n+m-w-m_1)  \label{eq:R1}\\
R_2 & :=  \sum_{w=0}^m  w \binom{m}{w} \Gamma (w+m_1-\alpha) \Gamma  (\theta+n+m-w-m_1) \label{eq:R2}\\
R_3 & :=  \sum_{w=0}^m   w^2  \binom{m}{w} \Gamma  (w+m_1-\alpha) \Gamma  (\theta+n+m-w-m_1)  \label{eq:R3}.
\end{align}
We focus on the evaluation of $R_1$:
\begin{align*}
R_1 & =  \Gamma (\theta+n+m-\alpha) \sum_{w=0}^m  \binom{m}{w}   B(w+m_1-\alpha , \theta+n+m-m_1-w)\\
& =  \Gamma (\theta+n+m-\alpha)   \int_0^1 \sum_{w=0}^m  \binom{m}{w} x^{w+m_1-\alpha-1} (1-x)^{\theta+n+m-m_1-w-1}\de x\\
& =   \Gamma (\theta+n+m-\alpha)  \int_0^1  x^{m_1-\alpha-1} (1-x)^{\theta+n+m-m_1-1} \Big( 1+\frac{x}{1-x}  \Big)  \de x\\
& =   \Gamma (\theta+n+m-\alpha) \int_0^1   x^{m_1-\alpha-1} (1-x)^{\theta+n-m_1-1} \de x   \\
& =  \Gamma (\theta+n+m-\alpha) B(m_1-\alpha , \theta+n-m_1),
\end{align*}
where we used the integral representation of the Beta function. The two sums \eqref{eq:R2}--\eqref{eq:R3} may be evaluated in a similar fashion to get
\begin{equation*}
R_2=  m \Gamma (\theta+n+m-\alpha) B (m_1+1-\alpha, \theta+n-m_1) 
\end{equation*}
and
\begin{equation*}
\begin{split}
R_3 &= m(m-1) \Gamma (\theta+n+m-\alpha) B (m_1-\alpha+2, \theta+n-m_1)\\
& \qquad\qquad +  m \Gamma (\theta+n+m-\alpha) B (m_1-\alpha+1 , \theta+n-m_1) .
\end{split}
\end{equation*}
Substituting the previous formulas for $R_1, R_2$ and $R_3$ in \eqref{eq:sommone}, the resulting expression can be used to evaluate the posterior expected value of $W_{1, n+m} \indic_{B_1}$ to get
\begin{equation} \label{eq:EW1B_final}
\begin{split}
&\E [W_{1, n+m} \indic_{B_1}| K_n =k , \{M_{i,n}=m_i\}_{i=1}^k] \\
 &\qquad =  \frac{n}{n+m} \cdot \frac{(\theta+1)_{(n-1)}}{[\alpha (n-m_1) +\theta m_1]  (\theta+1)_{(n+m-1)}}
 \cdot (\theta+n-\alpha)_{(m)}   C(n,m,\alpha, \theta,m_1)
\end{split}
\end{equation}
where 
\begin{equation*}
\begin{split}
& C(n,m,\alpha, \theta,m_1)  :=  m_1 [\alpha (n+m-m_1)+\theta m_1] +m\frac{m_1 -\alpha}{\theta+n-\alpha}[\alpha (n+m-m_1)+\theta m_1]\\
& \qquad +  \frac{m_1-\alpha}{\theta+n-\alpha} m (\theta-\alpha) (m_1+1) + \frac{(m_1-\alpha) (m_1+1-\alpha)}{(\theta+n-\alpha)(\theta+n+1-\alpha)} (\theta-\alpha) m (m-1) ,
\end{split}
\end{equation*}
and the thesis now follows putting together \eqref{eq:EWA1_final} and \eqref{eq:EW1B_final} to evaluate Equation (22).\\
\qed

\subsection{Auxiliary results}   \label{app:auxiliary}

We prove the following corollary of Theorem 1, which provides us with the posterior probabilities of the events $A_1$ and $B_1 $, defined in the main paper.
\begin{corollary}\label{cor:prob}
Let $((T_{1},X_{1}),\ldots,(T_{n},X_{n}))$ be a random sample under the BNP model (12), such that the sample features $K_{n}=k$ distinct species with corresponding ordered frequencies $\mathbf{M}_{n}=\mathbf{m}$. Then,
\begin{equation} \label{eq:PA1}
\text{Pr}  (A_1 | K_n =k ,\mathbf{M}_{n}=\mathbf{m})  =1-\frac{n}{n+m} \cdot \frac{(\theta+n+1-\alpha)_{(m)}}{(\theta+n)_{(m)}}
\end{equation}
and
\begin{equation} \label{eq:PB1}
\text{Pr}  (B_1 | K_n =k , \mathbf{M}_{n}=\mathbf{m} ) = \frac{n}{n+m} \cdot \frac{(\theta+n+1-\alpha)_{(m)}}{(\theta+n)_{(m)}}.
\end{equation}
\end{corollary}
\begin{proof}
The proof of Equation (19) and Equation (20) follows by a direct application of Theorem 1.
We now focus on the evaluation of Equation (21)
\eqref{eq:PB1}, and obviously \eqref{eq:PA1} follows from the fact that
\[
\text{Pr} [A_1 | K_n =k , \mathbf{M}_{n}=\mathbf{m}] = 1-\text{Pr} [B_1 | K_n =k , \mathbf{M}_{n}=\mathbf{m}] .
\]
In order to evaluate the probability on the r.h.s. of the previous expression, we consider Equation (20) and we sum all over the possible values of the random variable $W_{1,n+m}$:
\begin{align*}
&\text{Pr}[B_1 | K_n =k , \mathbf{M}_{n}=\mathbf{m}] \\
& \qquad = \sum_{w=0}^m  \binom{m}{w} \cdot \frac{n}{n+m} \cdot
\frac{\alpha (n+m-w-m_1)+\theta  (w+m_1)}{\alpha (n-m_1)+\theta m_1}\cdot\frac{(\theta+n-m_1)_{(m-w)} (m_1-\alpha )_{(w)}}{(\theta+n)_{(m)}}\\
& \qquad =  \frac{n}{(n+m) (\theta+n)_{(m)}  [\alpha (n-m_1)+\theta m_1] \Gamma (m_1-\alpha) \Gamma (\theta+n-m_1)}
(V_1+V_2),
\end{align*}
where $V_1$ and $V_2$ are defined as follows
\begin{align}
\label{eq:sumV1}
V_1 &:=  \sum_{w=0}^m  \binom{m}{w} [\alpha (n+m) +m_1 (\theta-\alpha)] \Gamma (m_1+w-\alpha) \Gamma (\theta+n+m-m_1-w)\\
\label{eq:sumV2}
V_2  & :=  \sum_{w=0}^m  \binom{m}{w} w  (\theta-\alpha) \Gamma (m_1+w-\alpha) \Gamma (\theta+n+m-m_1-w).
\end{align}
We focus on the evaluation of the two sums \eqref{eq:sumV1}--\eqref{eq:sumV2}. As for $V_1$ we get
\begin{equation}\label{eq:V1_valuated}
\begin{split}
V_1 & =  [\alpha (n+m) +m_1 (\theta-\alpha)] \sum_{w=0}^m  \binom{m}{w}  \Gamma (m_1+w-\alpha) \Gamma (\theta+n+m-m_1-w)\\
&  = [\alpha (n+m) +m_1 (\theta-\alpha)] \Gamma (\theta+n+m-\alpha) B(m_1-\alpha , \theta+n-m_1)
\end{split}
\end{equation}
where we have observed that the last sum coincides with $R_1$, that has been evaluated in Section \ref{app:expectation} of the Supplementary Material. The sum \eqref{eq:sumV2} equals
\begin{equation}
\label{eq:V2_evaluated}
\begin{split}
V_2 & = (\theta-\alpha) \sum_{w=0}^m  \binom{m}{w} w   \Gamma (m_1+w-\alpha) \Gamma (\theta+n+m-m_1-w) \\
& =  (\theta-\alpha)  m \Gamma (\theta+n+m-\alpha) B (m_1+1-\alpha, \theta+n-m_1)
\end{split}
\end{equation}
where we have now used the fact that the  sum coincides with $R_2$, which has been determined in Section \ref{app:expectation} of the Supplementary Material. We now substitute \eqref{eq:V1_valuated}--\eqref{eq:V2_evaluated} in the posterior probability of $B_1$, and with some simple calculations we get
\begin{align*}
\text{Pr}[B_1 | K_n =k , \mathbf{M}_{n}=\mathbf{m}] 
&  =  \frac{n \Gamma (\theta+n+m-\alpha) B(m_1-\alpha,  \theta+n-m_1)}{(n+m) (\theta+n)_m  [\alpha (n-m_1)+\theta m_1] \Gamma (m_1-\alpha) \Gamma (\theta+n-m_1)}\\
& \qquad \times   \Big\{ \alpha (n+m) +(\theta-\alpha)
\Big[ m_1 + m \frac{m_1-\alpha}{\theta+n-\alpha} \Big] \Big\}.
\end{align*}
The last expression in parenthesis equals
\[
\alpha (n+m) +(\theta-\alpha)
\Big[ m_1 + m \frac{m_1-\alpha}{\theta+n-\alpha} \Big] = (\theta+n+m-\alpha)\cdot \frac{\alpha (n-m_1)+\theta  m_1}{\theta+n-\alpha},
\]
as a consequence the posterior probability of $B_1$ boils down to 
\begin{align*}
\text{Pr}[B_1 | K_n =k , \mathbf{M}_{n}=\mathbf{m}] & =  \frac{n \Gamma (\theta+n+m-\alpha) B(m_1-\alpha,  \theta+n-m_1)}{(n+m) (\theta+n)_{(m)}  \Gamma (m_1-\alpha) \Gamma (\theta+n-m_1)}\cdot \frac{\theta+n+m-\alpha}{\theta+n-\alpha}\\
& = \frac{n}{n+m}\cdot \frac{(\theta+n+1-\alpha)_{(m)}}{(\theta+n)_{(m)}}
\end{align*}
where the last expression follows by a simple rearrangement of the terms, thus \eqref{eq:PB1} is now proved.
\end{proof}

\section{Additional empirical analyses} \label{app:application}

\subsection{Parameter estimation}

As described in section 3.2 of our main manuscript, parameters are either estimated using empirical Bayes approaches, or with a full Bayes approach.

Among the empirical Bayes approaches, we consider various methods that maximize the likelihood function. In this context, this means finding the optimal values of $\alpha$ and $\theta$ that maximize the EPPF.
Thanks to the available closed functional form, this can easily be done with an optimization routine, such as \texttt{optim} in R. 

We also consider a different kind of empirical Bayes approach, which instead of optimizing the full likelihood, considers a summary statistics and uses method of moments to match the empirical with the theoretical moments. Specifically, we consider the first theoretical moment and the observed value of the summary statistic, for a grid of sample size values. The summary statistics we consider are the number of distinct species and the frequency of the first ordered species; the pseudocode for these procedures is given respectively in algorithm \ref{alg:K} and algorithm \ref{alg:M1}.

These empirical Bayes algorithm are quite simple from a computational point of view, and they require on average 0.02 seconds to be computed on a MacBook Pro with a 3.1 GHz Dual-Core Intel Core i5 and 16GB LPDDR3 RAM. 

\begin{algorithm}
\caption{Empirical Bayes with method of moment for number of distinct species \label{alg:K}}
\begin{algorithmic}
\Procedure {$f_K$}{$\alpha,\theta;n$}
\State $EK \leftarrow \frac{\theta}{\alpha}(\frac{\Gamma(\theta+\alpha+n)}{\Gamma(\theta+\alpha)} \frac{\Gamma(\theta)}{\Gamma(\theta+n)} -1)$
\State \textbf{return} $EK$
\EndProcedure
\Statex
\Procedure {lsK}{$\{(T_1,X_1),\ldots,(T_n,X_n)\}$, $d$}
\For {$i$ in $1, \ldots, d$} 
    \State $n_i \leftarrow \textrm{floor}(i \cdot n/d)$ \Comment{Define equally spaced grid}
    \State $K_{n_i} \leftarrow \textrm{unique}(\{X_1,\ldots,X_{n_i}\})$ \Comment{Compute number of species up to $n_i$}
\EndFor
\State $(\hat \alpha,\hat \theta) = \argmin_{(\alpha,\theta)} \sum_{i=1}^d (K_{n_i} - f_{K}(\alpha,\theta; n_i))^2$ \Comment{Minimize least squares}
\EndProcedure
\end{algorithmic}
\end{algorithm}

\begin{algorithm}
\caption{Empirical Bayes with method of moment for first ordered species \label{alg:M1}}
\begin{algorithmic}
\Procedure {$f_{M_1}$}{$\alpha,\theta;n$}
\State $EM_1 \leftarrow \frac{\Gamma(\theta+1-\alpha+n)}{\Gamma(\theta+2-\alpha)} \frac{\Gamma(\theta+1)}{\Gamma(\theta+n)}$
\State \textbf{return} $EM_1$
\EndProcedure
\Statex
\Procedure {lsM1}{$\{(T_1,X_1),\ldots,(T_n,X_n)\}$, $d$}
\For {$i$ in $1, \ldots, d$} 
    \State $n_i \leftarrow \textrm{floor}(i \cdot n/d)$ \Comment{Define equally spaced grid}
    \State $i_{\max} \leftarrow \textrm{which.max}(\{T_1, \ldots, T_{n_i}\})$ \Comment{Find 1st ordered species up to $n_i$}
    \State $M_{1,n_i} \leftarrow \textrm{size}(\{X_1,\ldots,X_{n_i}\} = X_{i_{\max}})$ \Comment{Frequency of 1st ord. species up to $n_i$}
\EndFor
\State $(\hat \alpha,\hat \theta) = \argmin_{(\alpha,\theta)} \sum_{i=1}^d (M_{1,n_i} - f_{M_1}(\alpha,\theta; n_i))^2$ \Comment{Minimize least squares}
\EndProcedure
\end{algorithmic}
\end{algorithm}

For the full Bayes approach, a Metropolis Hasting algorithm is used to perform posterior inference on $\theta$ and $\alpha$. We iteratively update the value of $\theta$ and $\alpha$. We consider truncated normal distributions as proposals for both $\theta$ and $\alpha$. A burn-in of 500 iterations is discarded, and a total of 1000 samples are collected after thinning one every 50 iterations. This algorithm requires on average 16 seconds to run on the same MacBook Pro.

\subsection{Synthetic data evaluation: additional results} \label{app:additional}

In table~\ref{tab:additional} we report the performance of each parameter-estimation method, for each of the four quantities considered: the number of distinct species $K$, the frequency of the first ordered species $W_1$ and the frequency $W_1$ conditional on the events that the first ordered species is new ($A_1$) and that the first ordered species is old ($B_1$). We report the median percentage error, computed across 500 synthetic datasets, generated with different distributions: the ordered PYP model (\texttt{model}), or with a clustering distribution equal to either \texttt{DP},\texttt{PYP} or \texttt{zipf}, and with a ordering distribution equal to either \texttt{alpha-stable} or \texttt{arrival-weighted}.
For each one of 100 initial datasets of size $n = 500$, we consider the median percentage error across the 25 additional datasets of size $m = 5000$; we then report the median percentage error across those 100 median values. These results are the numeric correspondent to Figure 3 and 4 of the main manuscript.

\begin{table}
\centering
\begin{tabular}[t]{l|rrr|rrr|r}
\toprule
& \multicolumn{3}{c|}{alpha-stable} & \multicolumn{3}{c|}{arrival-weighted} &  \\
method & \multicolumn{1}{c}{DP} & \multicolumn{1}{c}{PYP} & \multicolumn{1}{c|}{zipf} & \multicolumn{1}{c}{DP} & \multicolumn{1}{c}{PYP} & \multicolumn{1}{c|}{zipf} & \multicolumn{1}{c}{model} \\
\midrule
\multicolumn{8}{l}{\textbf{K}}\\
\cellcolor{gray!6}{FB} & \cellcolor{gray!6}{0.481} & \cellcolor{gray!6}{0.284} & \cellcolor{gray!6}{0.095} & \cellcolor{gray!6}{0.299} & \cellcolor{gray!6}{0.126} & \cellcolor{gray!6}{0.132} & \cellcolor{gray!6}{0.092}\\
ordPYP & 0.388 & 0.238 & 0.099 & 0.195 & 0.126 & 0.151 & 0.096\\
\cellcolor{gray!6}{ordDP} & \cellcolor{gray!6}{0.096} & \cellcolor{gray!6}{0.327} & \cellcolor{gray!6}{0.606} & \cellcolor{gray!6}{0.094} & \cellcolor{gray!6}{0.332} & \cellcolor{gray!6}{0.608} & \cellcolor{gray!6}{0.328}\\
stdPYP & 0.098 & 0.131 & 0.105 & 0.101 & 0.129 & 0.109 & 0.113\\
\cellcolor{gray!6}{lsM1} & \cellcolor{gray!6}{3.798} & \cellcolor{gray!6}{1.626} & \cellcolor{gray!6}{0.935} & \cellcolor{gray!6}{0.713} & \cellcolor{gray!6}{0.519} & \cellcolor{gray!6}{0.720} & \cellcolor{gray!6}{0.750}\\
lsK & 0.127 & 0.233 & 0.117 & 0.125 & 0.226 & 0.118 & 0.179\\
\addlinespace
\multicolumn{8}{l}{\textbf{W1}}\\
\cellcolor{gray!6}{FB} & \cellcolor{gray!6}{0.628} & \cellcolor{gray!6}{1.237} & \cellcolor{gray!6}{1.718} & \cellcolor{gray!6}{0.562} & \cellcolor{gray!6}{0.778} & \cellcolor{gray!6}{1.425} & \cellcolor{gray!6}{1.121}\\
ordPYP & 0.617 & 1.269 & 1.671 & 0.530 & 0.776 & 1.493 & 1.232\\
\cellcolor{gray!6}{ordDP} & \cellcolor{gray!6}{0.962} & \cellcolor{gray!6}{11.487} & \cellcolor{gray!6}{19.028} & \cellcolor{gray!6}{0.707} & \cellcolor{gray!6}{4.469} & \cellcolor{gray!6}{16.609} & \cellcolor{gray!6}{8.635}\\
stdPYP & 0.675 & 2.195 & 1.742 & 0.556 & 0.816 & 1.259 & 1.114\\
\cellcolor{gray!6}{lsM1} & \cellcolor{gray!6}{0.708} & \cellcolor{gray!6}{0.856} & \cellcolor{gray!6}{0.896} & \cellcolor{gray!6}{0.567} & \cellcolor{gray!6}{0.874} & \cellcolor{gray!6}{0.883} & \cellcolor{gray!6}{0.650}\\
lsK & 0.642 & 2.376 & 1.962 & 0.640 & 0.847 & 1.289 & 1.097\\
\addlinespace
\multicolumn{8}{l}{\textbf{W1$\vert$A1}}\\
\cellcolor{gray!6}{FB} & \cellcolor{gray!6}{0.683} & \cellcolor{gray!6}{0.856} & \cellcolor{gray!6}{0.821} & \cellcolor{gray!6}{0.647} & \cellcolor{gray!6}{0.714} & \cellcolor{gray!6}{0.834} & \cellcolor{gray!6}{1.102}\\
ordPYP & 0.669 & 0.869 & 0.829 & 0.694 & 0.727 & 0.844 & 1.220\\
\cellcolor{gray!6}{ordDP} & \cellcolor{gray!6}{4.427} & \cellcolor{gray!6}{8.668} & \cellcolor{gray!6}{9.612} & \cellcolor{gray!6}{2.156} & \cellcolor{gray!6}{4.052} & \cellcolor{gray!6}{8.241} & \cellcolor{gray!6}{8.919}\\
stdPYP & 2.633 & 1.282 & 0.858 & 1.605 & 0.755 & 0.779 & 1.255\\
\cellcolor{gray!6}{lsM1} & \cellcolor{gray!6}{0.750} & \cellcolor{gray!6}{0.671} & \cellcolor{gray!6}{0.337} & \cellcolor{gray!6}{0.866} & \cellcolor{gray!6}{0.894} & \cellcolor{gray!6}{1.372} & \cellcolor{gray!6}{1.722}\\
lsK & 2.541 & 1.208 & 0.849 & 1.180 & 0.759 & 0.768 & 1.233\\
\addlinespace
\multicolumn{8}{l}{\textbf{W1$\vert$B1}}\\
\cellcolor{gray!6}{FB} & \cellcolor{gray!6}{0.270} & \cellcolor{gray!6}{0.352} & \cellcolor{gray!6}{0.505} & \cellcolor{gray!6}{0.109} & \cellcolor{gray!6}{0.199} & \cellcolor{gray!6}{0.271} & \cellcolor{gray!6}{0.259}\\
ordPYP & 0.271 & 0.357 & 0.503 & 0.109 & 0.200 & 0.270 & 0.266\\
\cellcolor{gray!6}{ordDP} & \cellcolor{gray!6}{0.330} & \cellcolor{gray!6}{0.473} & \cellcolor{gray!6}{1.294} & \cellcolor{gray!6}{0.110} & \cellcolor{gray!6}{0.196} & \cellcolor{gray!6}{0.297} & \cellcolor{gray!6}{0.321}\\
stdPYP & 0.291 & 0.361 & 0.490 & 0.110 & 0.201 & 0.263 & 0.266\\
\cellcolor{gray!6}{lsM1} & \cellcolor{gray!6}{0.285} & \cellcolor{gray!6}{0.353} & \cellcolor{gray!6}{0.500} & \cellcolor{gray!6}{0.108} & \cellcolor{gray!6}{0.195} & \cellcolor{gray!6}{0.209} & \cellcolor{gray!6}{0.262}\\
lsK & 0.282 & 0.355 & 0.526 & 0.109 & 0.199 & 0.268 & 0.278\\
\bottomrule
\end{tabular}
\caption{Median percentage error for simulations based on synthetic data, corresponding to the results displayed in Figure 3 and Figure 4. \label{tab:additional}}
\end{table}

\subsection{Genetic data evaluation: additional results} \label{app:additional_genetic}

We report here the plots showing the curve of $K$ and $W_1$ for the EDAR gene and the corresponding predictions and uncertainty quantification. Figure~\ref{fig:median_curves_EDAR}, which mirrors Figure~6, reports the curves (as black lines) of $K$ and $W_1$ for the training-testing split that represents the median error for a given parameter-estimating method, together with the corresponding predicted curves (depicted in red) and the confidence intervals (red bands). The confidence bands are created using bootstrap, computing the quantiles of the curve estimates across the training-testing splits. 

\begin{figure}[t!]
\centering
\includegraphics[width=0.97\linewidth]{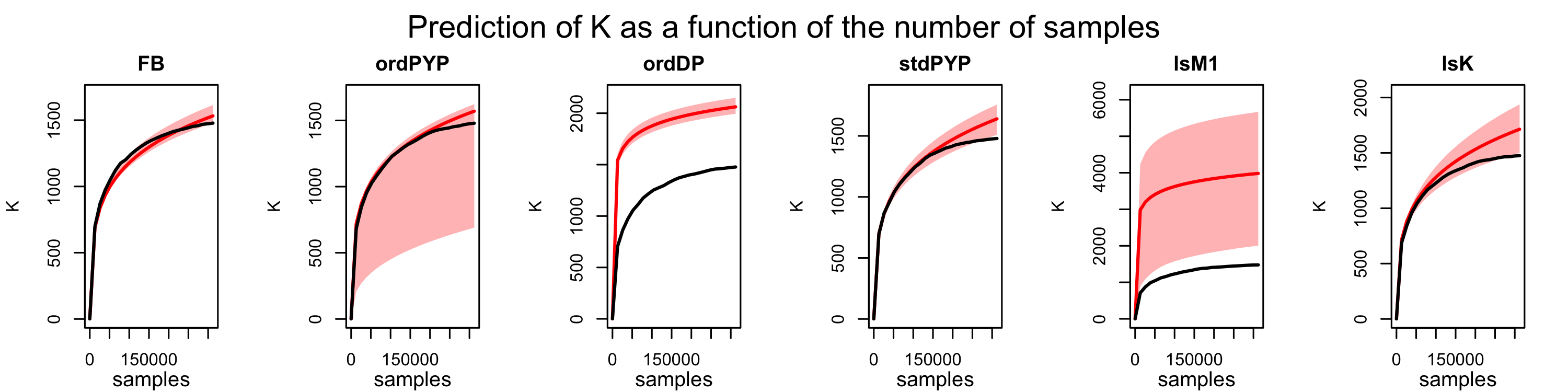} \\
\includegraphics[width=0.97\linewidth]{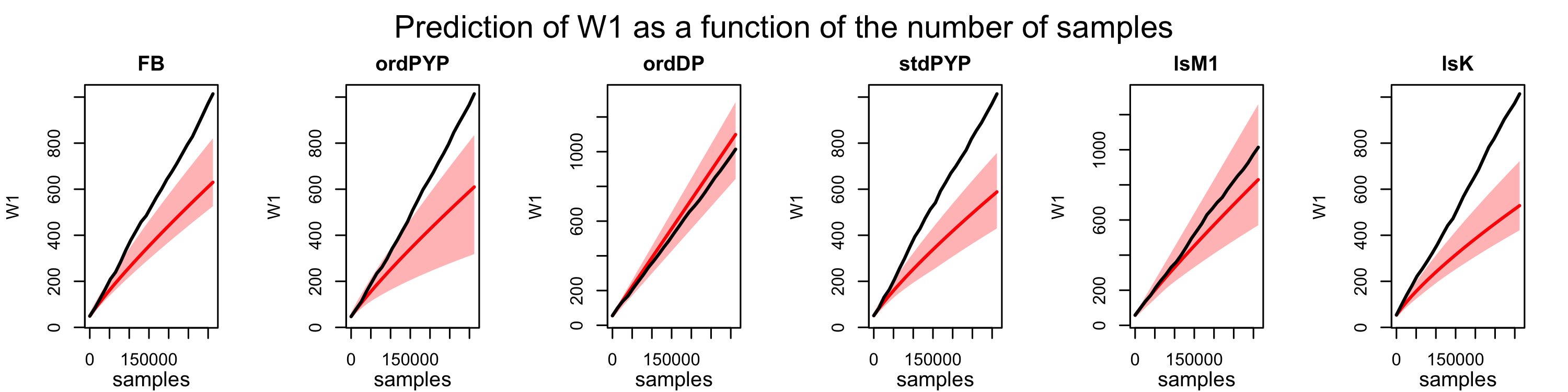} 
\caption{Prediction for the curve for the gene EDAR, for $K$ (top panels) and $W_1$ (bottom panels).} \label{fig:median_curves_EDAR}
\end{figure}

For the analysis of genetic variation in gene EDAR, Figure~\ref{fig:median_curves_EDAR} shows that the best performance in predicting the number of distinct variants is given by \texttt{stdPYP}, \texttt{ordPYP} and \texttt{FB}, while \texttt{ordDP} and \texttt{lsM1} better predict the frequency of the oldest variant. 

Moreover, table~\ref{tab:additional_genetic} reports the numeric results associated to Figure 5 of the main manuscript.

\begin{table}
\centering
\begin{tabular}[t]{l|rr|rr}
\toprule
& \multicolumn{2}{c|}{$\mathbf{K}$} & \multicolumn{2}{c}{$\mathbf{W1}$} \\
method & BRCA & EDAR & BRCA & EDAR\\
\midrule
\cellcolor{gray!6}{FB} & \cellcolor{gray!6}{0.167} & \cellcolor{gray!6}{0.043} & \cellcolor{gray!6}{0.023} & \cellcolor{gray!6}{0.142}\\
ordPYP & 0.162 & 0.053 & 0.017 & 0.143\\
\cellcolor{gray!6}{ordDP} & \cellcolor{gray!6}{0.280} & \cellcolor{gray!6}{0.394} & \cellcolor{gray!6}{0.004} & \cellcolor{gray!6}{0.007}\\
stdPYP & 0.138 & 0.100 & 0.261 & 0.182\\
\cellcolor{gray!6}{lsM1} & \cellcolor{gray!6}{0.216} & \cellcolor{gray!6}{1.665} & \cellcolor{gray!6}{0.008} & \cellcolor{gray!6}{0.018}\\
lsK & 0.153 & 0.162 & 0.293 & 0.220\\
\bottomrule
\end{tabular}
\caption{Median percentage error for the analysis of genetic data, corresponding to the results displayed in Figure 5. \label{tab:additional_genetic}}
\end{table}


\subsection{Analysis of citation data} \label{app:citations}

While ordered SSP were originally motivated by problems arising in population genetics, many other contexts can give rise to problems species sampling problems where species are associated to an order. One such context is the study of citation networks. In particular, given a collection of scientific articles, each citation reported in any of those articles can be seen as an individual observation, which can be classified into species based on which article is cited. Moreover, cited articles can be ordered based on their publication date.


In the following, we analyze citation data, using the dataset collected and cleaned by \cite{Ji(16)}, which focuses on papers published in the four top statistics journals from 2003 to 2012. The data is publicly available at \url{https://www.stat.uga.edu/directory/people/pengsheng-ji}, and detailed citation information (bibtex files, used to compute the ordering of the papers) where shared by the authors. 
Note that only the citations from articles in the dataset to articles in the dataset are considered.

\paragraph{Statistics citation data}

The Statistics Journals citation data provides the citation network of 3248 articles, of which 1798 contains citations to other papers in the dataset, and 1555 are cited by papers in the datasets, for a total of 5722 citations.
Within each journals, we ordered the papers by considering the year, the journal issue number, and the pages. To order the articles across journals we used the months for each issue, and when two journals published an issue in the same months, the order given by the pages was used. The order is increasing from the oldest paper to the most recent one.

We divide the original dataset into a training set and a testing set, so that the size of the former is one tenth of the original sample size. We repeat our analysis for 100 different randomly sampled training and testing sets. 
We focus our analysis on the prediction of the number of distinct species $K$ and the frequency of the first ordered species $W_1$. 
We compare the performance of the different parameter-estimation methods described in Sections~3.2. 

\begin{figure}[t!]
\centering
\includegraphics[width=0.6\linewidth]{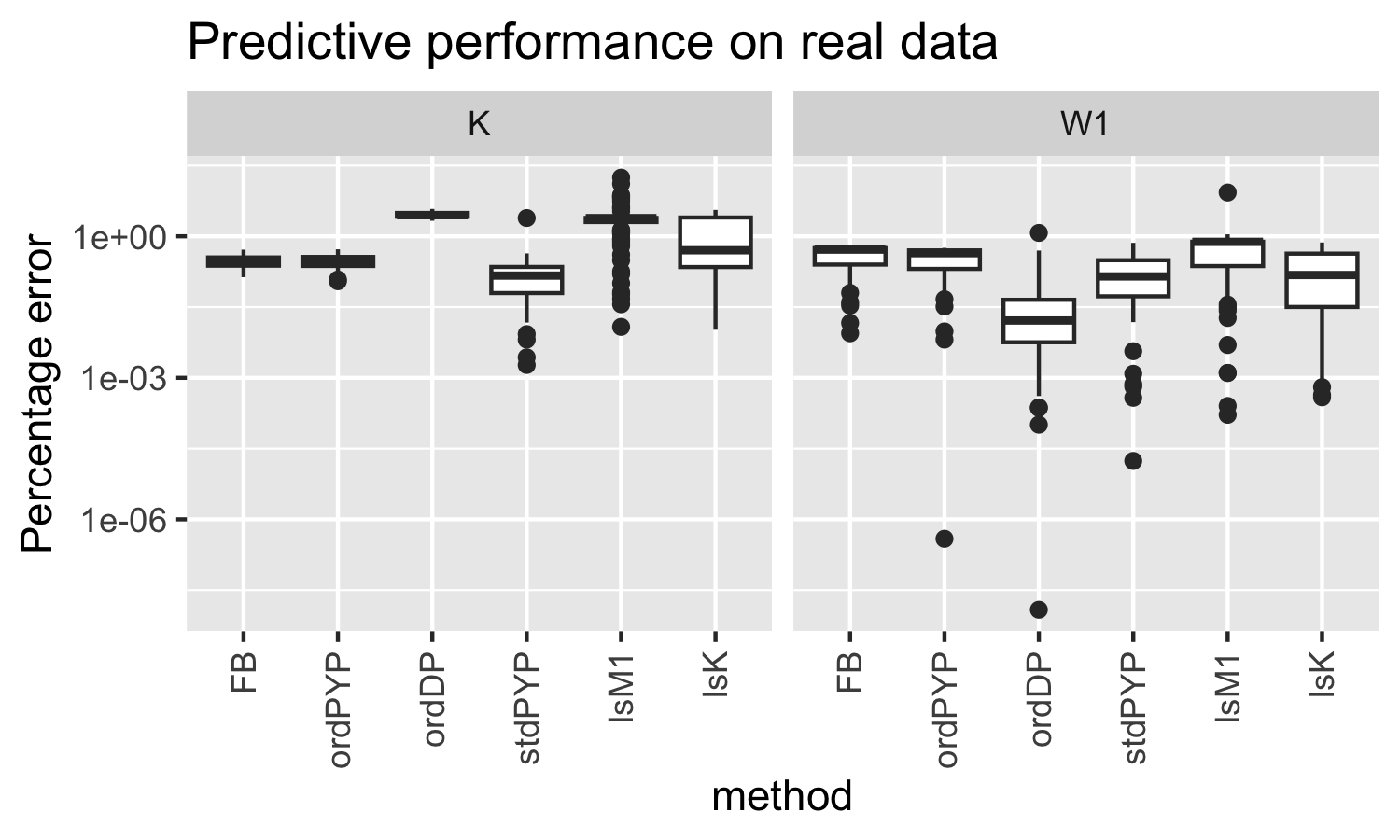}
\caption{Prediction performance across several training-testing sets for the statistics citation data. Note that results in the right panel are displayed using a log-transformed $y$ axis.} \label{fig:crossval_bibtex}
\end{figure}

Figure~\ref{fig:crossval_bibtex} reports the percentage absolute error for the prediction of $K$ and $W_1$ in the statistics citation data, computed across the several training-testing sets. We note that for the prediction of the number of distinct species $K$ the best performance is achieved by \texttt{stdPYP} with a median percentage error of $15\%$, followed by the \texttt{ordPYP} and \texttt{FB} with a median percentage error of $30\%$. 
For the prediction of the frequency of the oldest cited paper, the best performance is achieved by \texttt{ordDP}, with a median percentage error of $1\%$.
We also note that the percentage errors are quite smaller compared to those observed in Section~4.3, and this is due to the fact that in this dataset the event $A_1$, i.e. that the oldest citation is observed in the additional sample, is much more likely to occur.

\begin{figure}[t!]
\centering
\includegraphics[width=0.97\linewidth]{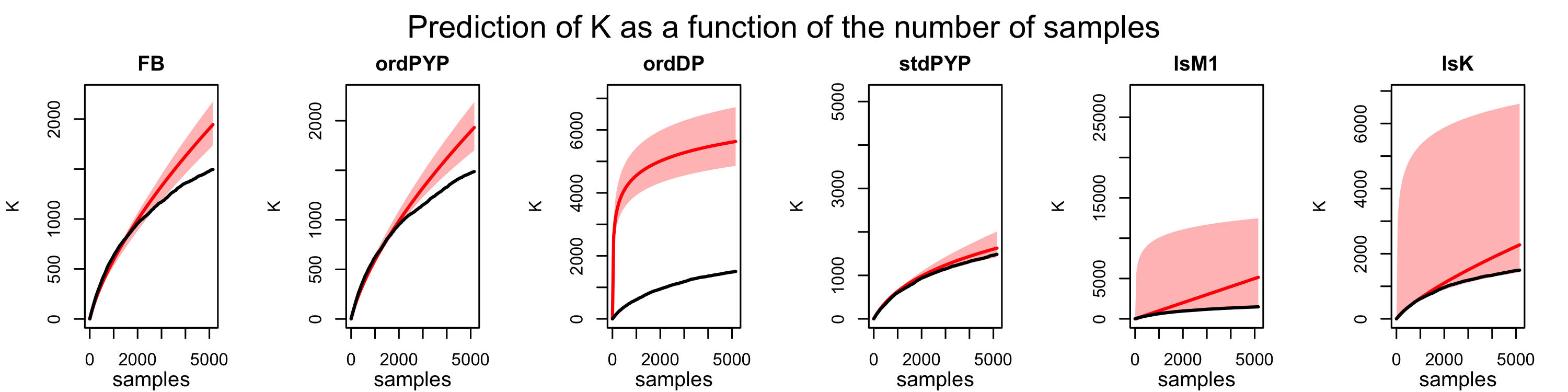} \\
\includegraphics[width=0.97\linewidth]{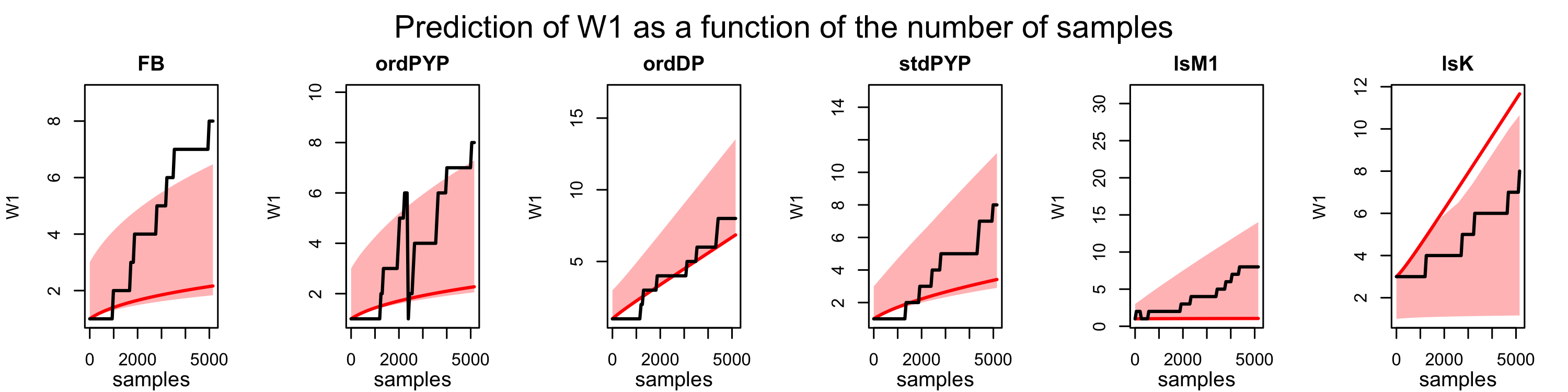} 
\caption{Prediction for the curve for the statistics citation data, for $K$ (top panels) and $W_1$ (bottom panels).} \label{fig:median_curves_bibtex}
\end{figure}

Similar behaviors are observed in Figure~\ref{fig:median_curves_bibtex}, which displays the predictions as functions of the number of samples, for the curve of the number of distinct species $K$ (top panels) and the frequency of the oldest species $W_1$ (bottom panels). 
For simplicity of visualization, we depict in black the actual curve of $K$ (or respectively $W_1$) observed in the training-testing split that most closely represent the median error achieved by each method. The prediction curve (represented in red) is also the one corresponding to the ``median'' training-testing split, while the red bands represent the $95\%$ confidence bands. 
We report the empirical confidence bands, computed using the empirical quantiles of the curve estimates across the various training-testing splits.
From the top panels of Figure~\ref{fig:median_curves_bibtex} it is clear that \texttt{stdPYP} provides a better fit to the whole curve of $K$, while the bottom panels show that \texttt{ordDP} has the best prediction for the curve of $W_1$. We also note that the curve of $W_1$ is not necessarily monotone increasing, as shown in the second and fourth bottom panels, because a new older cluster can be observed and reduce the frequency $W_1$ to one.

Overall, the parameters estimated with standard Pitman Yor Process seems to provide a better fit for the behavior of the number of distinct species (number of unique cited papers), but the parameters estimated with the ordered Dirichlet Process seems to better predict the frequency of the oldest cited paper $W_1$. This contrasting behavior is most likely explained by the fact that the ordered PYP model is somewhat misspecified for these data, which display power-law behavior for the number of species (consistent with the standard PYP), and a first ordered species frequency similar to that induced by the ordered DP.

\end{document}